\documentclass[reqno]{article}

\usepackage{amsthm}
\usepackage{amsmath}
\usepackage{amssymb}
\usepackage{graphicx}
\usepackage{color}
\usepackage{bbm}
\usepackage[top=2.5cm, bottom=2.5cm, left=2.5cm, right=2.5cm]{geometry}
\usepackage{enumerate}
\usepackage{faktor}
\usepackage{nicefrac}

\newcommand{\N}{\mathbb{N}}
\newcommand{\R}{\mathbb{R}}

\newcommand{\tr}{\mathrm{tr}}

\begin{document}

\numberwithin{equation}{section}
\newtheorem{theorem}[equation]{Theorem}
\newtheorem{remark}[equation]{Remark}
\newtheorem{claim}[equation]{Claim}
\newtheorem{lemma}[equation]{Lemma}
\newtheorem{definition}[equation]{Definition}
\newtheorem{assumptions}[equation]{Assumptions}
\newtheorem{corollary}[equation]{Corollary}
\newtheorem{proposition}[equation]{Proposition}
\newtheorem{def1}[equation]{Definition (First version)}
\newtheorem{def2}[equation]{Definition (Second version)}

\title{On the Existence of a Maximal Cauchy Development for the Einstein Equations - a Dezornification}
\author{Jan Sbierski\thanks{Department for Applied Mathematics and Theoretical Physics, University of Cambridge,
Wilberforce Road,
Cambridge,
CB3 0WA,
United Kingdom}}
\date{\today}

\maketitle

\begin{abstract}
In 1969, Choquet-Bruhat and Geroch established the existence of a unique maximal globally hyperbolic Cauchy development of given initial data for the Einstein equations. Their proof, however, has the unsatisfactory feature that it relies crucially on the axiom of choice in the form of Zorn's lemma. In this paper we present a proof that avoids the use of Zorn's lemma. In particular, we provide an explicit \emph{construction} of this maximal globally hyperbolic development.
\end{abstract}

\tableofcontents

\section{Introduction}

This paper is concerned with the initial value problem for the vacuum Einstein equations, \(Ric(g) =0\). In her seminal paper \cite{Choquet52} from 1952, Choquet-Bruhat showed that the initial value problem is \emph{locally} well-posed, i.e., in particular she proved a local existence and a local uniqueness statement. Global aspects of the Cauchy problem in general relativity were explored in the paper \cite{ChoquetGeroch69} by Choquet-Bruhat and Geroch from 1969, where they showed that for given initial data there exists a (unique) \emph{maximal} globally hyperbolic development (MGHD), i.e., a globally hyperbolic development  (GHD) which is an extension of any other GHD of the same initial data. The existence of the MGHD  not only implies `global uniqueness' for the Cauchy problem in general relativity within the class of globally hyperbolic developments, but it also defines the object whose properties one needs to understand for answering further questions about the initial value problem\footnote{Prominent and important examples are here the weak and the strong cosmic censorship conjectures, which are both concerned with the properties of the MGHD (for more details see Section \ref{Overview}).}  - thus turning the MGHD into a central object in mathematical general relativity. 

The proof of the existence of the MGHD, as given by Choquet-Bruhat and Geroch in \cite{ChoquetGeroch69}, has the unsatisfactory feature that it relies heavily on the axiom of choice in the form of Zorn's lemma, which they invoke in order to ensure the existence of such a maximal element \emph{without actually finding it}.
In this paper we present another proof of the existence of the MGHD which does not appeal to Zorn's lemma at all and, in fact, \emph{constructs} the MGHD.

\subsubsection*{Outline of the paper}

In the next subsection we elaborate more on the importance of the MGHD by discussing   the role it plays in the global theory of the Cauchy problem for the Einstein equations. Our motivation for giving another proof of the existence of the MGHD is discussed in Section \ref{Why}. Thereafter, we briefly recall the original proof by Choquet-Bruhat and Geroch. The impatient or knowledgeable reader is invited to skip directly to Section \ref{SketchMyProof}, where we sketch the idea of the proof given in this paper and exhibit the analogy of this new proof with the elementary proof of the existence of a unique MGHD for, say, a quasilinear wave equation on a fixed background manifold. Finally, Section \ref{Comparison} gives a brief schematic comparison of the original and the new proof.

In Section \ref{Defs}, we introduce the necessary definitions and state the main theorems, which are then proved in Section \ref{Proofs}.

\subsection{The maximal globally hyperbolic development in the global theory of the Cauchy problem in general relativity}
\label{Overview}

In the following we give a brief overview of the global aspects of the Cauchy problem in general relativity, focussing on the role played by the MGHD.  
Let us first discuss the aspect of `global uniqueness'.
In the paper \cite{ChoquetGeroch69}, Choquet-Bruhat and Geroch raised the following question:
\begin{quote}
A priori, it might appear possible that, once the solution has been integrated beyond a certain point in some region, the option, previously available, of further evolution in some quite different region has been destroyed\footnote{The possible scenario they describe here is well illustrated by the example of the simple ordinary differential equation \(\dot{x} = 3  x^{\nicefrac{2}{3}}\). If we prescribe, for instance, at time  \(t=-1\) the initial data \(x(t=-1) = -1\), then there is a \emph{unique solution up to time \(t=0\)}, given by \(x(t) = t^3\). At time \(t=0\), however, one can continue \(x\) as a \(C^2\) solution of the ODE in infinitely many ways, for example just by setting it to zero for all positive times.}.
\end{quote}
First of all it is clear, by looking at  the Kerr solution for example, that one can only hope to obtain a global uniqueness result if one restricts consideration to \emph{globally hyperbolic} developments of initial data\footnote{A globally hyperbolic development is not just a `development' which is globally hyperbolic, but one also requires that the initial data embeds as a Cauchy hypersurface. See Section \ref{Defs} for the precise definition of \emph{GHD}.}. That under this restriction, however, a global uniqueness statement indeed holds, was first proven in 1969 by Choquet-Bruhat and Geroch in the above cited paper. They actually proved a stronger statement than global uniqueness, namely they showed the existence of the MGHD, from which it follows trivially that global uniqueness holds. But the MGHD also furnishes the central object for the study of further global aspects of the Cauchy problem in general relativity. First and foremost one should mention here the \emph{weak} and the \emph{strong cosmic censorship conjectures}. The latter states that for \emph{generic} asymptotically flat initial data, that is data which models isolated gravitational systems, the MGHD cannot be isometrically embedded into a strictly larger spacetime (of a certain regularity). A positive resolution of the strong cosmic censorship conjecture would thus imply, that for asymptotically flat initial data, global uniqueness holds generically  \emph{even if we lift the restriction to globally hyperbolic developments}. 

We now come to the more subtle aspect of `global existence'. In fact, the sheer notion of a spacetime existing for `all time' is already non-trivial due to the absence of a fixed background manifold. However, the \emph{completeness of all causal geodesics} is a geometric invariant, which, moreover, accurately captures the physical concept of the spacetime existing for all time. And indeed, there are a few results which establish that global existence in this sense holds for small neighbourhoods of special initial data (see for example the monumental work of Christodoulou and Klainerman on the stability of Minkowski space, \cite{ChrisKlainStability}). On the other hand, there are explicit solutions to the Einstein equations which do not enjoy this causal geodesic completeness, showing that one cannot possibly hope to establish `global existence' in this sense for all initial data. Moreover, Penrose's famous singularity theorem, see \cite{Pen65}, shows that global existence in this sense cannot even hold generically\footnote{Penrose's singularity assumes that the development is globally hyperbolic, but recall from our discussion of global uniqueness, that this is the class of spacetimes we are interested in.}. 

If we restrict, however, our attention to asymptotically flat initial data, one could make the physically reasonable conjecture that at least the observers far out (at infinity) live for all time. Under the assumption that strong cosmic censorship holds, the mathematical equivalent of this physical conjecture is that null infinity of the corresponding MGHD is complete - which, for \emph{generic} asymptotically flat initial data, is the content of the weak cosmic censorship conjecture. Thus, the weak cosmic censorship conjecture should be thought of as conjecturing `global existence'.

\subsection{Why another proof?}
\label{Why}

Our motivation for giving another proof of the existence of the MGHD is mainly based on the following three arguments:

\begin{enumerate}[i)]

\item
A constructive proof is more natural and, from an epistemological point of view, more satisfying than a non-constructive one, since one can actually \emph{find} or \emph{construct} the object one seeks instead of inferring a contradiction by assuming its non-existence. Moreover, a direct construction usually provides not only more insight, but also more information.

\item
In his lecture notes \cite{Hilbert92}, David Hilbert distinguishes between two aspects of the mathematical method\footnote{For Hilbert's original words on this matter see \cite {Hilbert92}, page 17:  \scriptsize
\begin{quote}
Der Mathematik kommt hierbei eine zweifache Aufgabe zu: Einerseits gilt es, die Systeme von Relationen zu entwickeln und auf ihre logischen Konsequenzen zu untersuchen, wie dies ja in den rein mathematischen Disziplinen geschieht. Dies ist die \emph{progressive Aufgabe} der Mathematik. Andererseits kommt es darauf an, den an Hand der Erfahrung gebildeten Theorien ein festeres Gef\"uge und eine m\"oglichst einfache Grundlage zu geben. Hierzu ist es n\"otig, die Voraussetzungen deutlich herauszuarbeiten, und \"uberall genau zu unterscheiden, was Annahme und was logische Folgerung ist. Dadurch gewinnt man insbesondere auch Klarheit \"uber die unbewu{\ss}t gemachten Voraussetzungen, und man erkennt die Tragweite der verschiedenen Annahmen, so da{\ss} man \"ubersehen kann, was f\"ur Modifikationen sich ergeben, falls eine oder die andere von diesen Annahmen aufgehoben werden mu{\ss}. Dies ist die \emph{regressive Aufgabe} der Mathematik.
\end{quote}\normalsize}: He first mentions the \emph{progressive task} of mathematics, which is to establish a suitable set of postulates as the foundations of a theory, and then to investigate the theory itself by finding the logical consequences of its axioms. Hilbert then goes on to elaborate on the \emph{regressive task} of mathematics, which he says is to find and exhibit the logical dependency of the theorems on the postulates, which, in particular, leads to a clarification of the strength and the necessity of each axiom of the theory. 

The work in this paper is motivated by the regressive task, we show that the existence of the MGHD for the Einstein equations does \emph{not} rely on the axiom of choice. Besides a purely \emph{mathematical} motivation for investigating the strength and the necessity of each axiom of a theory, there is also an important \emph{physical} reason for doing so: The question whether an axiom or a theory is `true' is beyond the realm of mathematics. However, a \emph{physical} theory can be judged in accordance with its agreement with our perception of reality. For example, one would have a reason to dismiss the axiom of choice from the foundations of the \emph{physical} theory\footnote{Here, the `foundations of the physical theory' should be thought of as `mathematics with all its axioms together with those postulates within mathematics that actually model the physical theory'.}, if its inclusion in the remaining postulates of our physical theory allowed the deduction of a statement which is in serious disagreement with our perception of reality. On the other hand, it would be reasonable to include the axiom of choice in our axiomatic framework of the physical theory, if one could not prove a theorem, that is crucial for the physical theory, without it.

To the best of our knowledge, there are neither very strong arguments for embracing nor for rejecting the axiom of choice in general relativity. But if it had been the case that the axiom of choice had been needed for ensuring the existence of the MGHD, this would have been a strong reason for including it into the postulates of general relativity.

\item
The structure of the original proof of the existence of the MGHD is in stark contrast to the straightforward and elementary construction of the MGHD for, say, a quasilinear wave equation on a fixed background manifold; in the latter case one constructs the MGHD by taking the union of all GHDs (see also Section \ref{Quasilin}). The proof given in this paper embeds the construction of the MGHD for the Einstein equations in the general scheme for constructing MGHDs by showing that an analogous construction to `taking the union of all GHDs' works.
\end{enumerate}

We conclude with some formal set theoretic remarks: The results from PDE theory and causality theory we resort to in our proof do not require more choice than the \emph{axiom of dependent choice} (DC). Disregarding such `black box results' we refer to, our proof only needs the \emph{axiom of countable choice} (CC).\footnote{For one application of it see for example the proof of Lemma \ref{FuturePoint}.} We can thus conclude that the existence of the MGHD is a theorem of\footnote{\textbf{ZF} stands here for the Zermelo-Fraenkel set theory.} \textbf{ZF+DC}; and checking how much choice is actually required for proving the  `black box results' we resort to might even reveal that the existence of the MGHD is provable in \textbf{ZF+CC}. 

We have made no effort to avoid the axiom of \emph{countable} choice in our proof -  mainly for two reasons: Firstly, the axiom of countable choice is needed for many of the standard results and techniques in mathematical analysis. Thus, investigating whether the `black box results' we resort to can be proven even without the axiom of countable choice promises to be a rather tedious undertaking, while the gained insight might not be that enlightening.
Secondly, while the axiom of choice has rather wondrous consequences,  the implications of the axiom of countable (or dependent) choice seem, so far, to be less foreign to human intuition.

\subsection{Sketch of the proof given by Choquet-Bruhat and Geroch}

The original proof by Choquet-Bruhat and Geroch can be divided into two steps. In the first step, they invoke Zorn's lemma to ensure the existence of a \emph{maximal} element in the class of all developments; and in the second step, which is more difficult, they show that actually any other development embeds into this maximal element. Let us recall their proof in some more detail\footnote{The reader, who is not familiar with the terminology used below, is referred to the definitions made it Section \ref{Defs}.}:
\newline
\newline
\textbf{First step:} Consider the set \(\mathcal{M}\) of all globally hyperbolic developments of certain fixed initial data. Define a partial ordering on this set by \emph{\(M \leq M'\) iff \(M'\) is an extension of \(M\)}. Since a chain is by definition \emph{totally ordered}, it is not difficult to glue all the elements of a chain together\footnote{In particular it is trivial to show that the so obtained space is Hausdorff!} to construct a bound for the chain in question. Zorn's lemma then implies that there is at least one maximal element in  \(\mathcal{M}\). Pick one and call it \(M\).\footnote{The collection of all globally hyperbolic developments of given initial data is actually a proper class and not a set (see also footnote \ref{footclass} and the discussion above the proof of Theorem \ref{MGHD} on page \pageref{detailed}). In order to justify the above steps within \textbf{ZFC} (the Zermelo-Fraenkel set theory with the axiom of choice) one can perform a reduction to a \emph{set} $X$ of globally hyperbolic developments analogous to the reduction used in the proof of Theorem \ref{MGHD}. One then identifies elements in $X$ if, and only if, they are isometric (this is needed to ensure that $\leq$ as defined above is antisymmetric). The quotient space obtained in this way then takes the place of $\mathcal{M}$ in the above argument.}

\textbf{Second step:} Let \(M'\) be another element of \(\mathcal{M}\). Choquet-Bruhat and Geroch set up another partially ordered set, namely the set of all common globally hyperbolic developments of \(M\) and \(M'\), where the partial order is given by inclusion. Using the same argument as in Corollary \ref{welldefined}, they again argue that every chain is bounded, since one can just take the union of its elements. By appealing to Zorn's lemma once more, they establish the existence of a \emph{maximal} common globally hyperbolic development \(U\), and argue that it must be unique. 

Now, one glues  \(M\) and \(M'\) together along \(U\). The resulting space \(\tilde{M}\) can be endowed in a natural way with the structure needed for turning it into a globally hyperbolic development, which, however, might a priori be non-Hausdorff. Establishing that \(\tilde{M}\) is indeed Hausdorff is at the heart of their argument. Once this is shown, the resulting development is trivially an extension of \(M\) - and since \(M\) is maximal, we must have had \(U = M'\), i.e., \(M'\) embeds into \(M\).

The proof of \(\tilde{M}\) being Hausdorff goes by contradiction. If it were not Hausdorff, then one shows that this would be due to pairs of points on the boundary of \(U\) in \(M\) and \(M'\), respectively (cf. the picture below). One then has to ensure the existence of a `spacelike' part of this non-Hausdorff boundary. Given a `non-Hausdorff pair' \([p], [p'] \in \tilde{M}\), one then constructs a spacelike slice \(T\) in \(M\) that goes through \(p\) and such that \(T \setminus \{p\}\) is contained in \(U\). If \(\psi\) denotes the isometric embedding of \(U\) into \(M'\), this also gives rise to a spacelike slice \(T' := \psi\big(T \setminus \{p\}\big) \cup \{p'\}\) in \(M'\). 

\vspace{0.2cm}
\begin{center}
\def\svgwidth{14cm}
\begingroup%
  \makeatletter%
  \providecommand\color[2][]{%
    \errmessage{(Inkscape) Color is used for the text in Inkscape, but the package 'color.sty' is not loaded}%
    \renewcommand\color[2][]{}%
  }%
  \providecommand\transparent[1]{%
    \errmessage{(Inkscape) Transparency is used (non-zero) for the text in Inkscape, but the package 'transparent.sty' is not loaded}%
    \renewcommand\transparent[1]{}%
  }%
  \providecommand\rotatebox[2]{#2}%
  \ifx\svgwidth\undefined%
    \setlength{\unitlength}{536.07402344bp}%
    \ifx\svgscale\undefined%
      \relax%
    \else%
      \setlength{\unitlength}{\unitlength * \real{\svgscale}}%
    \fi%
  \else%
    \setlength{\unitlength}{\svgwidth}%
  \fi%
  \global\let\svgwidth\undefined%
  \global\let\svgscale\undefined%
  \makeatother%
  \begin{picture}(1,0.50476277)%
    \put(0,0){\includegraphics[width=\unitlength]{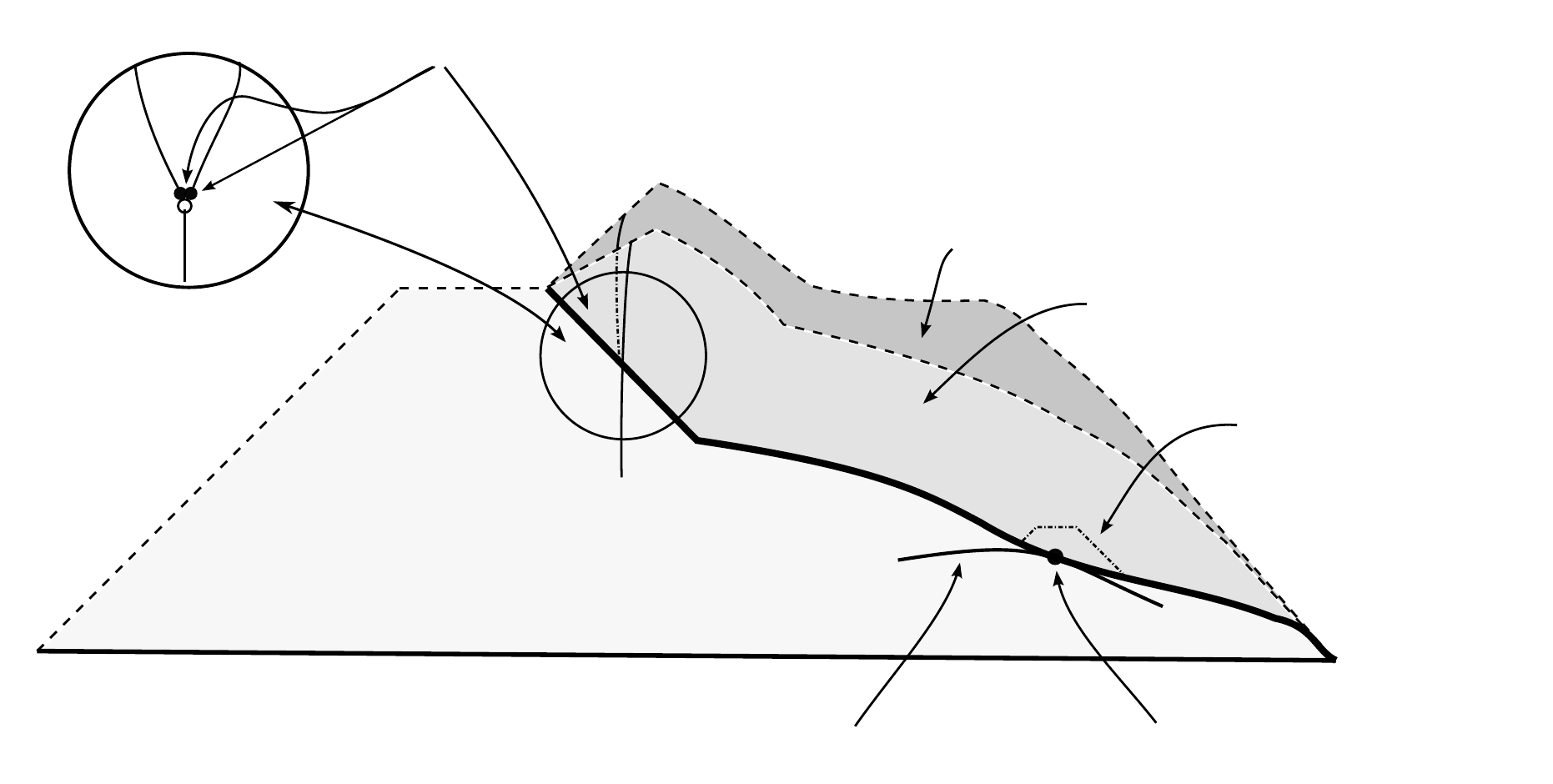}}%
    \put(0.19679939,0.47868777){\color[rgb]{0,0,0}\makebox(0,0)[lb]{\smash{Thick line contained twice \(=\) non-Hausdorff points }}}%
    \put(0.38293598,0.40073483){\color[rgb]{0,0,0}\makebox(0,0)[lb]{\smash{\(\tilde{M}\)}}}%
    \put(0.23730549,0.15996844){\color[rgb]{0,0,0}\makebox(0,0)[lb]{\smash{\(U\)}}}%
    \put(0.47288064,0.00753743){\color[rgb]{0,0,0}\makebox(0,0)[lb]{\smash{\(T\) and \(T'\)}}}%
    \put(0.70419202,0.01073523){\color[rgb]{0,0,0}\makebox(0,0)[lb]{\smash{\(p\) and \(p'\)}}}%
    \put(0.71283204,0.30386651){\color[rgb]{0,0,0}\makebox(0,0)[lb]{\smash{\(M\)}}}%
    \put(0.61232976,0.3516069){\color[rgb]{0,0,0}\makebox(0,0)[lb]{\smash{\(M'\)}}}%
    \put(0.80475305,0.22256329){\color[rgb]{0,0,0}\makebox(0,0)[lb]{\smash{extension of 
isometry}}}%
  \end{picture}%
\endgroup%

\end{center}

Clearly, the induced initial data on \(T\) and \(T'\) are isometric. Appealing to the local uniqueness statement for the initial value problem for the Einstein equations, one thus finds that one can actually extend the isometric identification of \(M\) with \(M'\) to a small neighbourhood of \(p\) - in contradiction with \(U\) being the \emph{maximal} common globally hyperbolic development. 

Let us remark that the proof of  \(\tilde{M}\) being Hausdorff is rather briefly presented in the original paper by Choquet-Bruhat and Geroch. A very detailed proof is found in Ringstr\"om's \cite{Ring13}.

\subsection{Outline of the proof presented in this paper}
\label{SketchMyProof}

We first discuss a proof of global uniqueness and of the existence of a MGHD for the case of a quasilinear wave equation on a fixed background manifold. Our proof for the case of the Einstein equations will then naturally arise by analogy. 

\subsubsection{The case of a quasilinear wave equation}
\label{Quasilin}

Let us consider a quasilinear wave equation for \(u : \R^{3+1} \to \R\),
\begin{equation}
\label{quasi}
g^{\mu \nu}(u, \partial u) \partial_\mu \partial_\nu u = F(u,\partial u)\;,
\end{equation}
where \(g\) is a Lorentz metric valued function. Under suitable conditions on \(g\) and \(F\) one can prove local existence and uniqueness of solutions to the Cauchy problem\footnote{We are not concerned with regularity questions here, all initial data can be assumed to be smooth.}. Such a statement takes the following form (see for example \cite{Sogge}):
\begin{equation}
\label{IV}
\begin{split}
&\textnormal{Given initial data \(f, h \in C_0^\infty(\R^3)\) there exists a \(T >0\) and a unique solution \(u \in C^\infty([0,T] \times \R^3)\)} \\ &\textnormal{of \eqref{quasi} with \(u(0, \cdot) = f(\cdot)\) and \(\partial_t u(0, \cdot) = h(\cdot)\). Moreover, if \(T^*\) denotes the supremum of all such} \\ &\textnormal{\(T>0\) then we have either \(T^* = \infty\) or the \(L^\infty(\R^3)\) norm of \(u (t,\cdot)\) and/or of some derivatives of} \\ &\textnormal{\(u\) blows up for \(t \to T^*\).}
\end{split}
\end{equation}
However, in the case of \(T^* < \infty\), in general \(u(x,t)\) will not become singular for all \(x\in \R^3\) for \(t \to T^*\). The points \(x \in \R^3\) where it becomes singular are called \emph{first singularities} - at regular spacetime points \((T^*,x)\) we can extend the solution. 

\begin{center}
\def\svgwidth{7cm}
\begingroup%
  \makeatletter%
  \providecommand\color[2][]{%
    \errmessage{(Inkscape) Color is used for the text in Inkscape, but the package 'color.sty' is not loaded}%
    \renewcommand\color[2][]{}%
  }%
  \providecommand\transparent[1]{%
    \errmessage{(Inkscape) Transparency is used (non-zero) for the text in Inkscape, but the package 'transparent.sty' is not loaded}%
    \renewcommand\transparent[1]{}%
  }%
  \providecommand\rotatebox[2]{#2}%
  \ifx\svgwidth\undefined%
    \setlength{\unitlength}{414.75429688bp}%
    \ifx\svgscale\undefined%
      \relax%
    \else%
      \setlength{\unitlength}{\unitlength * \real{\svgscale}}%
    \fi%
  \else%
    \setlength{\unitlength}{\svgwidth}%
  \fi%
  \global\let\svgwidth\undefined%
  \global\let\svgscale\undefined%
  \makeatother%
  \begin{picture}(1,0.70747141)%
    \put(0,0){\includegraphics[width=\unitlength]{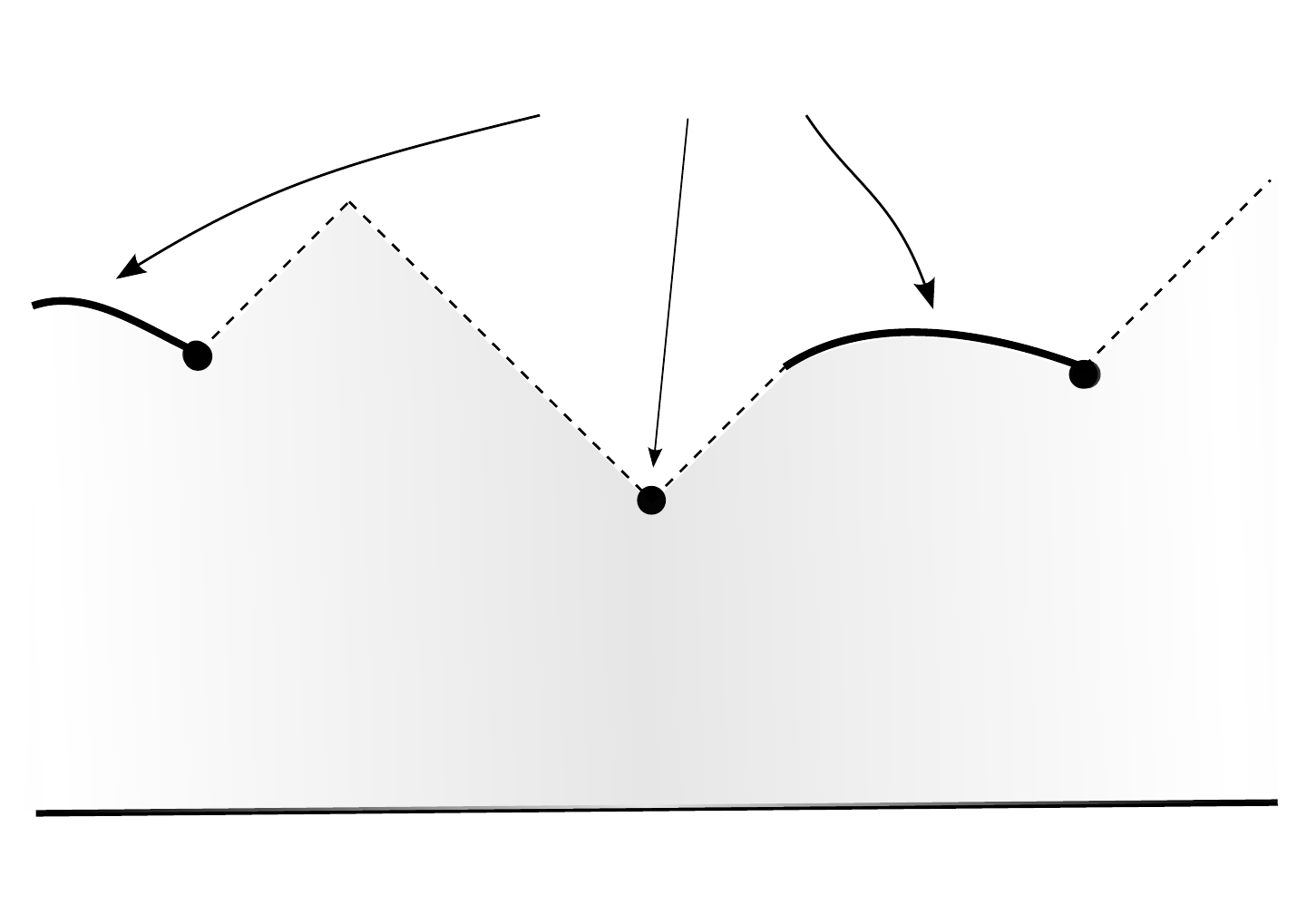}}%
    \put(0.36648209,0.64400009){\color[rgb]{0,0,0}\makebox(0,0)[lb]{\smash{First singularities}}}%
    \put(0.54283432,0.03227814){\color[rgb]{0,0,0}\makebox(0,0)[lb]{\smash{\(t=0\)}}}%
  \end{picture}%
\endgroup%

\end{center}

A natural question is then: does there exist a unique maximal globally hyperbolic\footnote{Note that it depends on the solution \(u\) whether a subset of \(\R^{3+1}\) is globally hyperbolic or not.} solution of \(\eqref{quasi}\) with initial values \(f\) and \(h\)? In the following we sketch a construction of such an object.
\newline
\newline
\textbf{First step:} We show that \emph{global uniqueness} holds, i.e., given two solutions \(u_1 : U_1 \to \R\) and \(u_2 : U_2 \to \R\) to the above Cauchy problem, where \(U_i\) is globally hyperbolic with respect to \(u_i\) and with Cauchy surface \(\{t=0\}\), the two solutions then agree on \(U_1 \cap U_2\). 

There are different ways to establish global uniqueness. One could for example prove this using energy estimates. Note, however, that such a proof is necessarily local by character, since \(U_1 \cap U_2\) is not a priori globally hyperbolic with respect to either of the solutions. 

The proof we sketch in the following is based on a continuity argument and only appeals to the local uniqueness statement. By this statement, we know that there is some open and globally hyperbolic neighbourhood \(V \subset U_1 \cap U_2\) of \(\{t=0\}\) on which the two solutions  agree (note that `global hyperbolicity' is here well-defined since the two solutions agree on the domain in question). Let us take the union \(W\) of all such \emph{common globally hyperbolic developments} (CGHD) of \((U_1, u_1)\) and \((U_2, u_2)\). By definition this set is clearly maximal, i.e., it is the biggest globally hyperbolic set on which \(u_1\) and \(u_2\) agree. We also call it the \emph{maximal common globally hyperbolic development} (MCGHD). 

Assume the so obtained set is not equal to \(U_1 \cap U_2\). Then, as in the picture below, we can take a small spacelike slice \(S\) that touches \(\partial W \cap U_1 \cap U_2\).\footnote{This step actually requires a bit of care...}
\vspace{3mm} 
\begin{center}
\def\svgwidth{8cm}
\begingroup%
  \makeatletter%
  \providecommand\color[2][]{%
    \errmessage{(Inkscape) Color is used for the text in Inkscape, but the package 'color.sty' is not loaded}%
    \renewcommand\color[2][]{}%
  }%
  \providecommand\transparent[1]{%
    \errmessage{(Inkscape) Transparency is used (non-zero) for the text in Inkscape, but the package 'transparent.sty' is not loaded}%
    \renewcommand\transparent[1]{}%
  }%
  \providecommand\rotatebox[2]{#2}%
  \ifx\svgwidth\undefined%
    \setlength{\unitlength}{461.92260742bp}%
    \ifx\svgscale\undefined%
      \relax%
    \else%
      \setlength{\unitlength}{\unitlength * \real{\svgscale}}%
    \fi%
  \else%
    \setlength{\unitlength}{\svgwidth}%
  \fi%
  \global\let\svgwidth\undefined%
  \global\let\svgscale\undefined%
  \makeatother%
  \begin{picture}(1,0.58580306)%
    \put(0,0){\includegraphics[width=\unitlength]{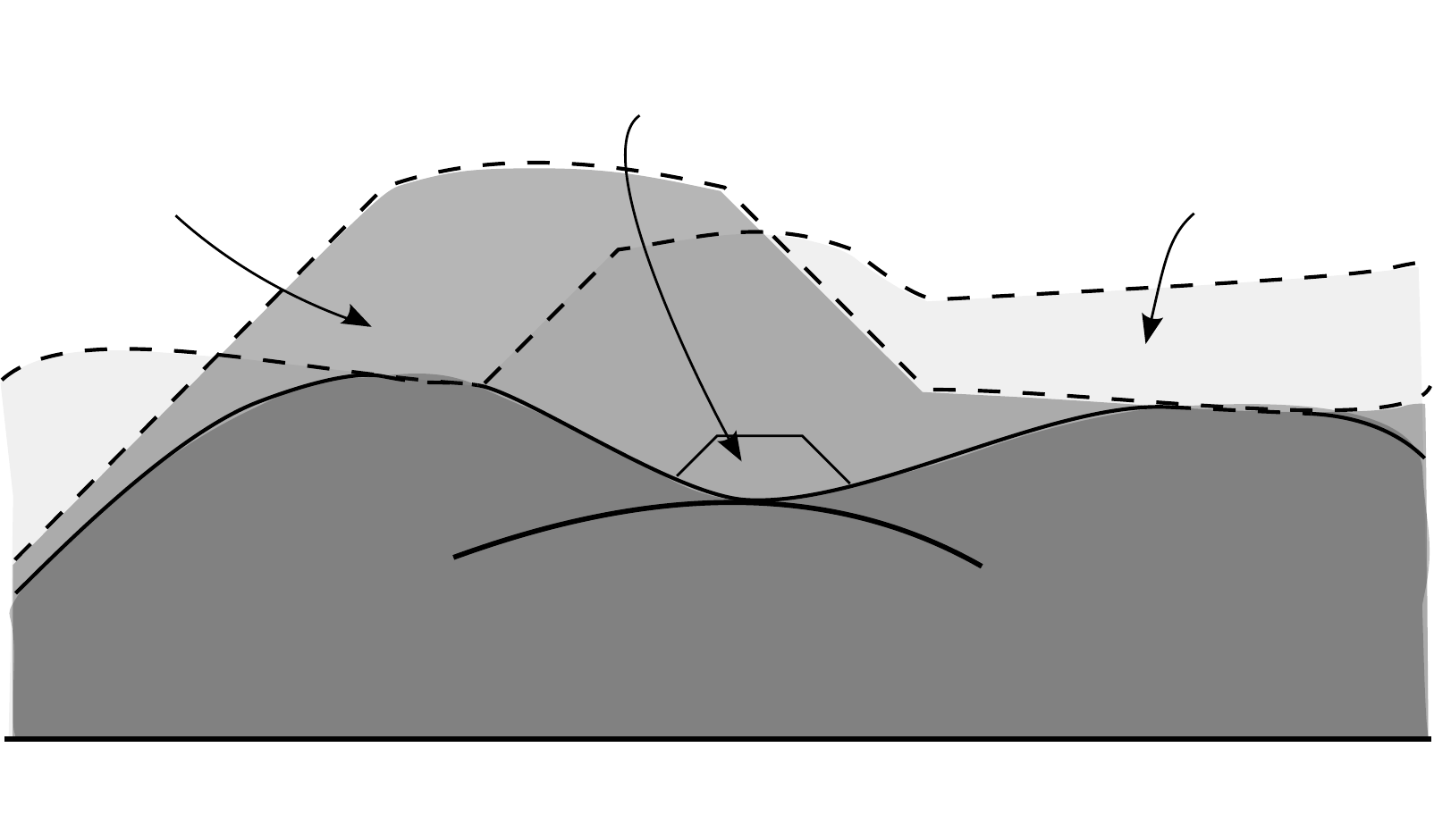}}%
    \put(0.6244145,0.01312111){\color[rgb]{0,0,0}\makebox(0,0)[lb]{\smash{\(t=0\)}}}%
    \put(0.07688775,0.13195232){\color[rgb]{0,0,0}\makebox(0,0)[lb]{\smash{\(W\)}}}%
    \put(0.06941163,0.45394102){\color[rgb]{0,0,0}\makebox(0,0)[lb]{\smash{\(U_1\)}}}%
    \put(0.84499414,0.44164321){\color[rgb]{0,0,0}\makebox(0,0)[lb]{\smash{\(U_2\)}}}%
    \put(0.69792376,0.17490681){\color[rgb]{0,0,0}\makebox(0,0)[lb]{\smash{\(S\)}}}%
    \put(0.24111765,0.52309302){\color[rgb]{0,0,0}\makebox(0,0)[lb]{\smash{Extension of MCGHD \(W\)}}}%
  \end{picture}%
\endgroup%

\end{center}

By assumption \(u_1\) and \(u_2\) agree in \(W\), thus by continuity they also agree on the slice \(S\). We now consider the initial value problem with the induced data on \(S\).\footnote{Note that a local uniqueness and existence statement for the initial value problem on \(S\) can be derived from \eqref{IV} by introducing slice coordinates for \(S\) and by appealing to the domain of dependence property.} Clearly, \(u_1\) and \(u_2\) are solutions, and thus, by the local uniqueness theorem, they agree in a small neighbourhood of \(S\). This, however, contradicts the maximality of \(W\). Hence, \(u_1\) and \(u_2\) agree on \(U_1 \cap U_2\).\footnote{The proof we just sketched yielded \(W = U_1 \cap U_2\) by \emph{contradiction}. However, it seems reasonable to expect that one can also prove \(W = U_1 \cap U_2\) \emph{directly} by the following continuity argument:
To begin with, the local uniqueness theorem shows that the set on which two solutions agree is not empty. By continuity of the solutions, we know then that the two solutions must also agree on the closure of this set, which furnishes the closedness part of the argument. Openness is achieved by restarting the local uniqueness argument from (spacelike slices that touch) the boundary, as in the above picture. Note however, that in order to obtain openness across null boundaries, one has to ``work one's way upwards'' along the null boundary, which makes this direct argument a bit more complicated. Also note that this continuity argument is qualitatively the same as the one already encountered in proving uniqueness of solutions to the initial value problem for regular ODEs.}

\textbf{Second step:} Having proved global uniqueness, the construction of the MGHD is now a trivial task:
We consider the set of all globally hyperbolic developments \(\{U_\alpha,u_\alpha\}_{\alpha \in A}\) of the initial data \(f\), \(h\) and note that this set is non-empty by the local existence theorem. We then take the union \(U := \bigcup_{\alpha \in A} U_\alpha\) of all the domains \(U_\alpha\) and define 
\begin{equation*}
u(x):= u_\alpha(x) \textnormal{ for } x \in U_\alpha\;.
\end{equation*}
By global uniqueness, this is well-defined. Moreover, it is easy to see  that the set \(U\) is globally hyperbolic with respect to \(u\) and that this development is maximal by construction.

\subsubsection{The case of the Einstein equations}

Our proof of the existence of the MGHD for the Einstein equations can be viewed as an `imitation' of the scheme just presented. To understand better the problems that have to be overcome, however, let us first qualitatively compare the Einstein equations with a quasilinear wave equation on a fixed background manifold: A solution to the Einstein equations is given by a pair \((M,g)\), where \(M\) is a manifold and \(g\) a Lorentzian metric on \(M\). The background manifold \(M\) is \emph{not} fixed here. The diffeomorphism invariance of the Einstein equations states that if \(\phi\) is a diffeomorphism from \(M\) to a manifold \(N\), then \((N,\phi_*g)\) is also a solution to the Einstein equations. Physically, these two solutions are indistinguishable - which suggests that one should consider the Einstein equations as `equations for isometry classes of Lorentzian manifolds' (cf. also Remark \ref{RemSym}). It is also only then that the Einstein equations become hyperbolic. Moreover, it is well-known that breaking the diffeomorphism invariance by imposing a harmonic gauge (this should be thought of as picking a representative of the isometry class) turns the Einstein equations into a system of quasilinear wave equations. 
It is thus reasonable to expect that the only problems caused in transferring the construction of the MGHD from Section \ref{Quasilin} to the Einstein equations are due to the fact that, while in the case of the quasilinear wave equation the objects one works with are functions defined on subsets of a fixed background manifold, for the Einstein equations one actually would have to consider isometry classes of Lorentzian manifolds. In particular we face the following two problems:
\begin{enumerate}[i)]
\item Already the \emph{definition} of `global uniqueness' does not transfer directly to the Einstein equations, since \(U_1 \cap U_2\) is not a priori defined for two GHDs \(U_1\) and \(U_2\) for the Einstein equations.
\item Since there is no fixed ambient space in the context of the Einstein equations, one cannot just take the union of all GHDs of given initial data in order to construct the MGHD.
\end{enumerate}

We discuss the first problem first.  
For the case of the quasilinear wave equation on a fixed background manifold, a trivially equivalent formulation of `global uniqueness' is that there is a globally hyperbolic development \((U,u)\) of the initial data such that \(U_1 \cup U_2\) is contained in \(U\) and such that \(u = u_1\) on \(U_1\) and \(u=u_2\) on \(U_2\). 

This formulation of `global uniqueness' \emph{does} transfer to the Einstein equations: \emph{Given two globally hyperbolic developments of the same initial data, there exists a globally hyperbolic development in which both isometrically embed}. This statement is the content of Theorem \ref{CommonExtension}. Moreover, it is exactly this notion of global uniqueness that is crucial for the existence of the MGHD.

Let us first motivate the method used in this paper for constructing this common extension of two GHDs for the Einstein equations: In the case of a quasilinear wave equation on a fixed background manifold, we would construct a common extension of \((U_1,u_1)\) and \((U_2,u_2)\) by first showing that the solutions agree on \(U_1 \cap U_2\) - as we did in Section \ref{Quasilin} - and thereafter extending both solutions to \(U_1 \cup U_2\). Let us observe here that instead of constructing the bigger space \(U_1 \cup U_2\) by taking the union of \(U_1\) and \(U_2\), we can also glue them together along \(U_1 \cap U_2\) - which yields the same result.
However, for the construction of the common extension, both operations only make sense, if we already know that the solutions agree on \(U_1 \cap U_2\). We can, however, still \emph{glue} along an a priori smaller set on which we know that the two solutions agree, i.e., along a common globally hyperbolic development \(V\) of \(U_1\) and \(U_2\). In general, the so obtained space will not be Hausdorff due to the presence of `corresponding boundary points', i.e., a point in \(\partial V \) that lies in \(U_1\) as well as in \(U_2\). The same argument which established global uniqueness above (cf. the last picture) shows, however, that if this is the case, then we can actually find a bigger CGHD along which we can glue. 

Let us now directly glue \(U_1\) and \(U_2\) together along the \emph{maximal} CGHD (recall, that this was defined as the union of all CGHDs). Again, the same argument that corresponds to the last picture shows that the MCGHD of \((U_1, u_1)\) and \((U_2, u_2)\) \emph{cannot have corresponding boundary points}\footnote{In particular we inferred that thus the MCGHD must be equal to \(U_1 \cap U_2\).} since this would violate the maximality of the MCGHD. In particular, we see that glueing along the MCGHD yields a Hausdorff space. 

This reinterpretation of the construction of the common extension \(U_1 \cup U_2\) of \(U_1\) and \(U_2\) for the case of a quasilinear wave equation can be transferred to the Einstein equations: In Section \ref{ExistenceMCGHD} we establish the existence of the MCGHD for two given GHDs for the Einstein equations. Note that this is also proved in the original paper by Choquet-Bruhat and Geroch - however, they appeal to Zorn's lemma. Here, we construct the MCGHD of two GHDs \(U_1\) and \(U_2\) by taking the union of all CGHDs (that are subsets of \(U_1\)) in \(U_1\).    
In Section \ref{PropertyMCGHD} we then give the rigorous proof that the MCGHD does not have corresponding boundary points, i.e., that the space obtained by glueing along the MCGHD, lets call it \(\tilde{M}\), is then indeed Hausdorff.  Moreover, it is more or less straightforward to show that \(\tilde{M}\) satisfies all other properties of a GHD, see Section \ref{Final}, which then finishes the construction of the common extension and thus proves global uniqueness for the Einstein equations. 

Let us summarise the main idea that guided the way for the construction of the common extension of two GHDs for the Einstein equations:
\begin{equation}
\label{NewIdeaEinstein}
\begin{split}
\textnormal{In the case of the Einstein equations, the appropriate analogue of `taking the union'} \\ \textnormal{of two GHDs is to glue them together along their MCGHD.} \qquad \qquad \qquad
\end{split}
\end{equation}

This statement, in spite of its simplicity, should be considered as the main new idea of this paper. It also leads straightforwardly to the construction of the MGHD in the case of the Einstein equations by proceeding in analogy to the case of a quasilinear wave equation on a fixed background manifold: for given initial data, we glue `all' GHDs together along their MCGHDs, see Section \ref{Final}.\footnote{\label{footclass}Let us already  remark here the following subtlety: The collection of \emph{all} GHDs of given initial data forms a \emph{proper} class, i.e., it is too `large' for being a set and, hence, also for performing the glueing construction using the axioms of \textbf{ZF}. In Section \ref{Final} we show that it suffices to work with an appropriate subclass of all GHDs, which actually is a \emph{set}.}

\subsection{Schematic comparison of the two proofs}
\label{Comparison}

\renewcommand{\arraystretch}{1.7}

\begin{center}
  \begin{tabular}{ p{7cm}| p{7cm} }
    \hspace{2cm}\textbf{Original proof} & \hspace{2.5cm}\textbf{New proof} \\ \hline
    Ensure existence of a maximal element \(M\) in the set of all GHDs (using Zorn's lemma).  &  \\
    Ensure existence of a MCGHD of two GHDs (using Zorn's lemma). & Construct MCGHD of two GHDs by taking the union (literally!) of all CGHDs. \\
    & Prove global uniqueness by `taking the union' (in the sense of \eqref{NewIdeaEinstein})  of two GHDs. \\
    Show that \(M\) is indeed the MGHD by `taking the union' (in the sense of \eqref{NewIdeaEinstein}).  & Construct MGHD by `taking the union' (in the sense of \eqref{NewIdeaEinstein}) of `all' GHDs. \\
    (Infer global uniqueness from the existence of the MGHD.) & \\ 
  \end{tabular}
\end{center}

\renewcommand{\arraystretch}{1.0}

\section{The basic definitions and the main theorems}
\label{Defs}

Let us start with some words about the stipulations we make: 
\begin{itemize}
\item
This paper is only concerned with the smooth case, i.e., we only consider smooth initial data for the Einstein equations. In particular, the MGHD we construct is, a priori, only maximal among \emph{smooth} GHDs. This raises the question whether one could extend the MGHD to a bigger GHD that is, however, less regular. 

An answer to this question is provided by the \emph{low regularity} local well-posedness theory for quasilinear wave equations, which in particular entails that \emph{as long as the solution remains in the low regularity class local well-posedness is proven in, any additional regularity is preserved}. The classical approach using energy estimates yields such a local well-posedness statement for very general quasilinear wave equations in \(H^{\nicefrac{5}{2} + \varepsilon}\). For the special case of the Einstein equations, the recent resolution of the bounded \(L^2\) curvature conjecture by Klainerman, Rodnianski and Szeftel implies that additional regularity is preserved as long as (roughly speaking) the metric is in $H^2$ (see \cite{KlRodSzef12} for details).

Regarding the technique of the proof given in this paper, it heavily depends on the causality theory developed for at least \(C^2\)-regular Lorentzian metrics. But as long as the initial data is such that it gives rise to a GHD of regularity at least \(C^2\), basically the same proof as given in this paper goes through. For work on the existence of the MGHD for rougher initial data along the lines of the original Choquet-Bruhat Geroch style argument using Zorn's lemma, see \cite{Chrus11}. Here one should mention that up to a few years ago the proof of local uniqueness (which plays, not surprisingly, a central role for proving global uniqueness) required one degree of differentiability more than the proof of local existence. This issue was overcome by Planchon and Rodnianski (\cite{PlaRod07}).

Having made these comments, we stipulate that all manifolds and tensor fields considered in this paper are smooth, even if this is not mentioned explicitly.

\item
We moreover assume that all Lorentzian manifolds we consider are connected and time oriented. The dimension of the Lorentzian manifolds is denoted by \(d+1\), where \(d \geq 1\).
\item
For simplicity of exposition we restrict our consideration to the \emph{vacuum} Einstein equations \(Ric(g) =0\). However, the inclusion of matter and/or of a cosmological constant does not pose any additional difficulty as long as a local existence and uniqueness statement holds. In fact, exactly the same proof applies.
\end{itemize}

The Einstein equations are of hyperbolic character, they allow for a well-posed initial value problem. \emph{Initial data} \((\overline{M},\bar{g},\bar{k})\) for the  vacuum Einstein equations consists of a connected \(d\)-dimensional Riemannian manifold \((\overline{M},\bar{g})\) together with a symmetric \(2\)-covariant tensor field \(\bar{k}\) on \(\overline{M}\) that satisfy the \emph{constraint equations}:
\begin{align*}
\bar{R} - |\bar{k}|^2 + (\tr{\bar{k}})^2 &=0 \\
\bar{\nabla}^i \bar{k}_{ij} - \bar{\nabla}_j\tr{\bar{k}} &=0 \;,
\end{align*}
where \(\bar{R}\) denotes the scalar curvature  and \(\bar{\nabla}\) denotes the Levi-Civita connection on \(\overline{M}\). 

\begin{definition}
A \emph{globally hyperbolic development} (GHD) \((M,g,\iota)\) of initial data \((\overline{M},\bar{g},\bar{k})\) is a time oriented, globally hyperbolic Lorentzian manifold \((M,g)\) that satisfies the vacuum Einstein equations, together with an embedding \(\iota : \overline{M} \to M\) such that 
\begin{enumerate}
\item \(\iota^*(g) = \bar{g}\)
\item \(\iota^*(k) = \bar{k}\), where \(k\) denotes the second fundamental form of \(\iota(\overline{M})\) in \(M\) with respect to the future normal
\item \(\iota(\overline{M})\) is a Cauchy surface in \((M,g)\).
\end{enumerate}
\end{definition}

\begin{definition}
\label{extension}
Given two GHDs \((M,g,\iota)\) and \((M',g',\iota')\) of the same initial data, we say that \((M',g',\iota')\) is an \emph{extension} of \((M,g,\iota)\) iff there exists a time orientation preserving isometric embedding\footnote{We lay down some terminology here: An \emph{isometry} is a diffeomorphism that preserves the metric. An \emph{isometric immersion} is an immersion that preserves the metric. Finally, an \emph{isometric embedding} is an isometric immersion that is a diffeomorphism onto its image.} \(\psi : M \to M'\) that preserves the initial data, i.e.\ \(\psi \circ \iota = \iota'\). 
\end{definition}

\begin{def1}
\label{PreDefCGHD}
Given two GHDs \((M,g,\iota)\) and \((M',g',\iota')\) of  initial data \((\overline{M},\bar{g},\bar{k})\), we say that a GHD \((U,g_U,\iota_U)\) of the same initial data is a \emph{common globally hyperbolic development} (CGHD) of \((M,g,\iota)\) and \((M',g',\iota')\) iff both \((M,g,\iota)\) and \((M',g',\iota')\) are extensions of \((U,g_U,\iota_U)\).
\end{def1}

Paraphrasing Definition \ref{PreDefCGHD}, a GHD \(U\) is a CGHD of GHDs \(M\) and \(M'\) if, and only if, \(U\) is `contained' in \(M\) as well as in \(M'\). Here we have just written \(M\) instead of \((M,g,\iota)\), etc. We will from now on often use this shorthand notation.

We now give a slightly different definition of a common globally hyperbolic development and discuss the relation with the previous definition thereafter in Remark \ref{RemSym}.

\begin{def2}
\label{DefCGHD}
Given two GHDs \((M,g,\iota)\) and \((M',g',\iota')\) of  initial data \((\overline{M},\bar{g},\bar{k})\), we say that a GHD \((U \subseteq M,g|_U,\iota)\) is a \emph{common globally hyperbolic development} (CGHD) of \((M,g,\iota)\) and \((M',g',\iota')\) iff  \((M',g',\iota')\) is an extension of \((U,g_U,\iota_U)\).
\end{def2}

\begin{remark}
\label{RemSym}
\begin{enumerate}
\item
The diffeomorphism invariance of the Einstein equations implies that if \(M\) is a GHD of certain initial data, then so is any spacetime that is isometric to \(M\). From a physical point of view, isometric spacetimes should be considered to be the same, i.e., one should actually consider the \emph{isometry class} of a GHD to be the solution to the Einstein equations. It is easy to check that the \emph{Definitions \ref{extension} and \ref{PreDefCGHD} also descend to the isometry classes of GHDs, i.e., they do not depend on the chosen representative of the isometry class}. It is also only when one considers isometry classes that one can prove uniqueness for the initial value problem to the Einstein equations in the strict meaning of this word. However, working with isometry classes has a decisive disadvantage for the purposes of this paper: the isometry class of a given GHD is a \emph{proper} class, i.e., not a set. Thus, if we considered an infinite number of isometry classes, not even the full axiom of choice would be strong enough to pick a representative of each - and we need a representative to work with. We thus refrain from considering isometry classes of GHDs. 

\item
As just mentioned, Definition \ref{PreDefCGHD} is diffeomorphism invariant. In Definition \ref{DefCGHD} we break the diffeomorphism invariance by requiring that a CGHD \(U\) of \(M\) and \(M'\) is realised as a subset of \(M\). However, this is not a serious restriction, since given any CGHD \(U\) of \(M\) and \(M'\) in the sense of Definition \ref{PreDefCGHD}, we can isometrically embed \(U\) into \(M\) by using the isometric embedding that is provided by \(M\) being an extension of \(U\).

Although Definition \ref{DefCGHD} is a bit less natural, we will choose it over Definition \ref{PreDefCGHD} in this paper since, for our purposes, it is more convenient to work with. Also note that while Definition \ref{PreDefCGHD} is symmetric in \(M\) and \(M'\), i.e., \emph{\(U\) being a CGHD of \(M\) and \(M'\)} is the same as \emph{\(U\) being a CGHD of \(M'\) and \(M\)}, the symmetry is broken in Definition \ref{DefCGHD}.
\end{enumerate}
\end{remark}

The local existence and uniqueness theorem for the initial value problem for the vacuum Einstein equations can now be phrased as follows:

\begin{theorem}
\label{LocalTheory}
Given initial data for the vacuum Einstein equations, there exists a GHD, and for any two GHDs of the same initial data, there exists a CGHD.
\end{theorem}

The essential details of this theorem were proven by Choquet-Bruhat in 1952, see \cite{Choquet52}.
The next two theorems are the main theorems of this paper.

\begin{theorem}[Global uniqueness]
\label{CommonExtension}
Given two GHDs \(M\) and \(M'\) of the same initial data, there exists a GHD \(\tilde{M}\) that is an extension of \(M\) and \(M'\).
\end{theorem}

\begin{theorem}[Existence of MGHD]
\label{MGHD}
Given initial data there exists a GHD \(\tilde{M}\) that is an extension of any other GHD of the same initial data.  The GHD \(\tilde{M}\) is unique up to isometry and is called the \emph{maximal globally hyperbolic development (MGHD)} of the given initial data.
\end{theorem}

Note that Theorem \ref{MGHD} clearly implies Theorem \ref{CommonExtension}. In the original proof by Choquet-Bruhat and Geroch, Theorem \ref{MGHD} was proven without first proving Theorem \ref{CommonExtension}. In our approach, however, we first establish Theorem \ref{CommonExtension} - Theorem \ref{MGHD} then follows easily.

\section{Proving the main theorems}
\label{Proofs}

\subsection{The existence of the maximal common globally hyperbolic development}
\label{ExistenceMCGHD}

In this section we construct the unique maximal common globally hyperbolic development of two GHDs. We start with a couple of lemmata that are needed for this construction.

\begin{lemma}
\label{IsoConnected}
Let \((M,g)\) and \((M',g')\) be Lorentzian manifolds, where \(M\) is connected. Furthermore, let \(\psi_1, \psi_2 : M \to M'\) be two isometric immersions with \(\psi_1(p) = \psi_2(p)\) and \(d\psi_1 (p) = d\psi_2(p)\) for some \(p \in M\). It then follows that \(\psi_1 = \psi_2\).
\end{lemma}

\begin{proof}
One shows that the set 
\begin{equation*}
A= \big{\{} x \in M \; | \; \psi_1(x) = \psi_2(x) \textnormal{ and } d\psi_1(x) = d\psi_2(x) \big{\}}
\end{equation*}
is open, closed and non-empty, from which it then follows that \(A = M\). In order to show openness, let \(x_0 \in A\) be given and choose a normal neighbourhood \(U\) of \(x_0\). For \(x \in U\), there is then a geodesic \(\gamma : [0, \varepsilon] \to U\) with \(\gamma (0) = x_0\) and \(\gamma(\varepsilon) = x\). Since \(\psi_1\) and \(\psi_2\) are both isometric immersions, we have that both \(\psi_1 \circ \gamma\) and \(\psi_2 \circ \gamma\) are geodesics. Moreover, since by assumption we have \((\psi_1\circ \gamma)(0) = (\psi_2\circ \gamma)(0)\) and \(\dot{(\psi_1 \circ \gamma)}(0) = \dot{(\psi_2 \circ \gamma)}(0)\), the two geodesics agree. In particular, we obtain \(\psi_1(x) = (\psi_1 \circ \gamma) (\varepsilon) = (\psi_2 \circ \gamma) ( \varepsilon) = \psi_2(x)\). 

The closedness of \(A\) follows from the smoothness of \(\psi_1\) and \(\psi_2\), and non-emptyness holds by assumption.
\end{proof}

\begin{corollary}
\label{welldefined}
Let \((M,g)\) be a globally hyperbolic, time oriented Lorentzian manifold with Cauchy surface \(\Sigma\) and \((M',g')\) another time oriented Lorentzian manifold. Moreover, say \(U_1, U_2 \subseteq M\) are open and globally hyperbolic with Cauchy surface \(\Sigma\), and \(\psi_i : U_i \to M'\), \(i=1,2\), are time orientation preserving isometric immersions that agree on \(\Sigma\).

Then \(\psi_1\) and \(\psi_2\) agree on \(U_1 \cap U_2\).
\end{corollary}

\begin{proof}
Since \(\psi_1\) and \(\psi_2\) agree on \(\Sigma\), their differentials agree on \(\Sigma\) if evaluated on vectors tangent to \(\Sigma\). Moreover, since the isometric immersion preserve the time orientation, they both map the future normal of \(\Sigma\) onto the future normal of \(\psi_1(\Sigma) = \psi_2(\Sigma)\). Thus, the differentials of \(\psi_1\) and \(\psi_2\) agree on \(\Sigma\). The corollary now follows from Lemma \ref{IsoConnected}.
\end{proof}

\begin{lemma}
\label{IsoDiffeo}
Say \((M,g)\) and \((M',g')\) are two globally hyperbolic spacetimes with Cauchy surfaces \(\Sigma\) and \(\Sigma'\), respectively. Let \(\psi : M \to M'\) be an isometric immersion such that \(\psi|_\Sigma : \Sigma \to \Sigma'\) is a diffeomorphism.

Then \(\psi\) is an isometric embedding.
\end{lemma}

Note that this shows in particular that in Definition \ref{extension} one does not need to require \(\psi\) to be an isometric embedding - \(\psi\) being an isometric immersion suffices.

\begin{proof}
It suffices to show that \(\psi\) is injective. So let \(p,q\) be points in \(M\) with \(\psi(p) = \psi(q)\). Consider an inextendible timelike geodesic \(\gamma : (a,b) \to M\) with \(\gamma(0) = q\), where \(-\infty \leq a < 0 < b \leq \infty\). Since \((M,g)\) is globally hyperbolic, \(\gamma\) intersects \(\Sigma\) exactly once; say \(\gamma(\tau_0) \in \Sigma\), where \(\tau_0 \in (a,b)\). Note that since \(\psi\) is an isometric immersion, \(\psi \circ \gamma : (a,b) \to M'\) is also a timelike geodesic. We now choose a neighbourhood \(V\) of \(p\) such that \(\psi\big|_V : V \to \psi(V)\) is a diffeomorphism and we pull back the velocity vector of \(\psi \circ \gamma\) at \(\psi(p)\) to \(p\). Let \(\sigma : (c,d) \to M\) denote the inextendible timelike geodesic with \(\sigma(0) =p\) and \(\dot{\sigma}(0) = d\psi\big|_V^{-1}\big(\dot{\psi \circ \gamma}\big)\big|_{\psi(q)}\), where \(-\infty \leq c < 0 < d \leq \infty\). Again, by \(M\) being globally hyperbolic, \(\sigma\) intersects \(\Sigma\) exactly once; say at \(\sigma(\tau_1) \in \Sigma\), whith \(c < \tau_1 < d\). Clearly, the geodesics \(\psi \circ \gamma\) and \(\psi \circ \sigma\) agree on their common domain, since they share the same initial data.

\begin{center}
\def\svgwidth{15cm}
\begingroup%
  \makeatletter%
  \providecommand\color[2][]{%
    \errmessage{(Inkscape) Color is used for the text in Inkscape, but the package 'color.sty' is not loaded}%
    \renewcommand\color[2][]{}%
  }%
  \providecommand\transparent[1]{%
    \errmessage{(Inkscape) Transparency is used (non-zero) for the text in Inkscape, but the package 'transparent.sty' is not loaded}%
    \renewcommand\transparent[1]{}%
  }%
  \providecommand\rotatebox[2]{#2}%
  \ifx\svgwidth\undefined%
    \setlength{\unitlength}{574.81279297bp}%
    \ifx\svgscale\undefined%
      \relax%
    \else%
      \setlength{\unitlength}{\unitlength * \real{\svgscale}}%
    \fi%
  \else%
    \setlength{\unitlength}{\svgwidth}%
  \fi%
  \global\let\svgwidth\undefined%
  \global\let\svgscale\undefined%
  \makeatother%
  \begin{picture}(1,0.29382422)%
    \put(0,0){\includegraphics[width=\unitlength]{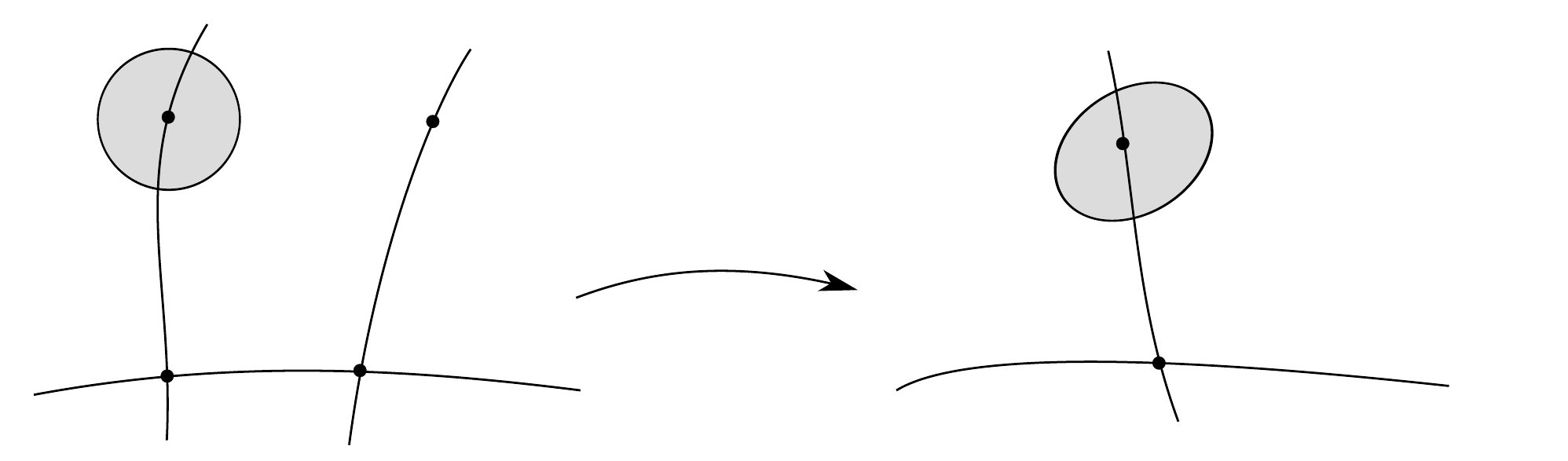}}%
    \put(0.11857847,0.23035362){\color[rgb]{0,0,0}\makebox(0,0)[lb]{\smash{\(p\)}}}%
    \put(0.28563786,0.2064243){\color[rgb]{0,0,0}\makebox(0,0)[lb]{\smash{\(q\)}}}%
    \put(0.43209135,0.13756505){\color[rgb]{0,0,0}\makebox(0,0)[lb]{\smash{\(\psi\)}}}%
    \put(0.72170415,0.21207712){\color[rgb]{0,0,0}\makebox(0,0)[lb]{\smash{\(\psi(p) = \psi(q)\)}}}%
    \put(0.10763113,0.1199002){\color[rgb]{0,0,0}\makebox(0,0)[lb]{\smash{\(\sigma\)}}}%
    \put(0.259579,0.13610066){\color[rgb]{0,0,0}\makebox(0,0)[lb]{\smash{\(\gamma\)}}}%
    \put(0.29431441,0.02509411){\color[rgb]{0,0,0}\makebox(0,0)[lb]{\smash{\(\Sigma\)}}}%
    \put(0.77653374,0.02931169){\color[rgb]{0,0,0}\makebox(0,0)[lb]{\smash{\(\Sigma'\)}}}%
    \put(0.04968999,0.24863016){\color[rgb]{0,0,0}\makebox(0,0)[lb]{\smash{\(V\)}}}%
    \put(0.62891558,0.23035366){\color[rgb]{0,0,0}\makebox(0,0)[lb]{\smash{\(\psi(V)\)}}}%
    \put(0.73131623,0.10947663){\color[rgb]{0,0,0}\makebox(0,0)[lb]{\smash{\(\psi \circ \gamma\)}}}%
  \end{picture}%
\endgroup%

\end{center}

By the global hyperbolicity of \((M',g')\), the geodesics \(\psi \circ \gamma\) and \(\psi \circ \sigma\) cannot intersect \(\Sigma'\) more than once, which implies that \(\tau_0 = \tau_1\). Moreover, since \(\psi\big|_\Sigma : \Sigma \to \Sigma'\) is a diffeomorphism, we have \(\sigma(\tau_0) = \gamma(\tau_0)\). Now making use again of \(\psi\) being a local diffeomorphism at \(\sigma(\tau_0)\), one infers that \(\dot{\sigma}(\tau_0) = \dot{\gamma}(\tau_0)\) also holds. It follows that \(\sigma = \gamma\) and in particular that \(p = \sigma(0) = \gamma(0) = q\).
\end{proof}

We can finally prove the main result of this section:

\begin{theorem}[Existence of MCGHD]
\label{ExMCGHD}
Given two GHDs \(M\) and \(M'\) of the same initial data, there exists a unique  CGHD \(U\) of \(M\) and \(M'\) with the property that if \(V\) is another CGHD of \(M\) and \(M'\), then \(U\) is an extension of \(V\). 

We call \(U\) the \emph{maximal common globally hyperbolic development (MCGHD) of \(M\) and \(M'\)}.
\end{theorem}

The original proof of this theorem, i.e., as it is found in \cite{ChoquetGeroch69} or \cite{Ring13} for example, appeals to Zorn's lemma. The much simpler method of taking the union of all CGHDs of \(M\) and \(M'\) however works:

\begin{proof}
We consider the set \(\{U_\alpha \subseteq M \;\big|\; \alpha \in A\}\) of all CGHDs of \(M\) and \(M'\). By Theorem \ref{LocalTheory} this set is non-empty. We show that 
\begin{equation*}
U:= \bigcup_{\alpha \in A} U_\alpha
\end{equation*}
is the MCGHD of \(M\) and \(M'\).
\begin{enumerate}
\item It is clear that \(U\) is open an thus a time-oriented Ricci-flat Lorentzian manifold.
\item \(U\) is globally hyperbolic with Cauchy surface \(\iota(\overline{M})\): Let \(\gamma\) be an inextendible timelike curve in \(U\). Take a point on \(\gamma\); it lies in some \(U_\alpha\) and the corresponding curve segment in \(U_\alpha\) can be considered to be an inextendible timelike curve in \(U_\alpha\) and thus has to meet \(\iota(\overline{M})\). Note that \(\gamma\) cannot meet \(\iota(\overline{M})\) more than once, since \(\gamma\) is also a segment of an inextendible timelike curve in \(M\) - and \(M\) is globally hyperbolic.
\item It follows that \(U\) is a GHD of the given initial data.
\item \(U\) is a CGHD of \(M\) and \(M'\): It suffices to give an isometric immersion \(\psi : U \to M'\) that respects the embedding of \(\overline{M}\) and the time orientation. Note that by Lemma \ref{IsoDiffeo} \(\psi\) is then automatically an isometric embedding.

For each \(\alpha \in A\) there is such an isometric immersion \(\psi_\alpha : U_\alpha \to M'\). We define
\begin{equation*}
\psi(p) := \psi_\alpha(p) \qquad \textnormal{ for } p \in U_\alpha \,.
\end{equation*}
By Corollary \ref{welldefined} this is well-defined and clearly \(\psi\) is an isometric immersion that respects the embedding of \(\overline{M}\) and the time orientation.
\item That \(U\) is maximal follows directly from its definition. It then also follows that \(U\) is the unique CGHD with this maximality property. 
\end{enumerate}
\end{proof}

\subsection{The maximal common globally hyperbolic development does not have corresponding boundary points}
\label{PropertyMCGHD}

In this section we prove that the MCGHD of two GHDs \(M\) and \(M'\) does not have `corresponding boundary points'. Most of the proofs found in this section are based on proofs found in Chapter 23 of Ringstr\"om's book \cite{Ring13}.

\begin{definition}
\label{DefBound}
Let \(U\) be a CGHD of \(M\) and \(M'\), and let us denote the isometric embedding of \(U\) into \(M'\) with \(\psi\). Two points \(p \in \partial U \subseteq M\) and \(p' \in \partial \psi(U) \subseteq M'\) are called \emph{corresponding boundary points} of \(U\) iff for all neighbourhoods \(V\) of \(p\) and for all neighbourhoods \(V'\) of \(p'\) one has
\begin{equation*}
\psi^{-1}\big( V' \cap \psi(U)\big) \cap V \neq \emptyset\;.
\end{equation*}
\end{definition}

The main theorem of this section is

\begin{theorem}
\label{NotMCGHD}
Let \(M\) and \(M'\) be GHDs of the same initial data, and say \(U\) is a CGHD of \(M\) and \(M'\). If there are corresponding boundary points of \(U\) in \(M\) and \(M'\), then there exists a strictly larger extension of \(U\) that is also a CGHD of \(M\) and \(M'\). In particular, \(U\) is not the MCGHD of \(M\) and \(M'\).
\end{theorem}

Before we give the proof of Theorem \ref{NotMCGHD}, we need to establish some results concerning the structure and properties of corresponding boundary points. Let us begin by giving a different characterisation of corresponding boundary points using timelike curves, which will often prove more convenient.

\begin{proposition}
\label{EquivCha}
Let \(U\) be a CGHD of \(M\) and \(M'\) with isometric embedding \(\psi : U \subseteq M \to M'\).
The following statements are equivalent:
\begin{enumerate}[i)]
\item The points \(p \in \partial U\) and \(p' \in \partial \psi(U)\) are corresponding boundary points.
\item If \(\gamma : (-\varepsilon, 0) \to U\) is a timelike curve with \(\lim_{s \nearrow 0} \gamma(s) = p\), then \(\lim_{s \nearrow 0}  (\psi \circ \gamma) (s) =p'\).
\item There is a timelike curve \(\gamma : (-\varepsilon, 0) \to U\) with \(\lim_{s \nearrow 0} \gamma(s) = p\) such that \(\lim_{s \nearrow 0}  \psi \circ \gamma (s) =p'\).
\end{enumerate}
In particular it follows from \(ii)\) and \(iii)\) that \(p \in \partial U\) has at most one corresponding boundary point.
\end{proposition}

Before we give the proof, let us recall some notation from causality theory on time oriented Lorentzian manifolds\footnote{For a detailed discussion of causality theory on Lorentzian manifolds the reader is referred to Chapter \(14\) of \cite{ONeill}.}: we write 
\begin{enumerate}
\item \(p  \ll q\) iff there is a future directed timelike curve from \(p\) to \(q\)
\item \(p < q\) iff there is a future directed causal curve from \(p\) to \(q\)
\item \(p \leq q\) iff \(p < q\) or \(p = q\).
\end{enumerate} 

\begin{proof}[Proof of Proposition \ref{EquivCha}:]
The implications \(ii) \implies iii)\) and \(iii) \implies i)\) are trivial. We prove \(i) \implies ii)\): Without loss of generality let us assume that \(p\) and \(p'\) lie to the future of the Cauchy surfaces \(\iota(\overline{M})\) and \(\iota'(\overline{M})\), respectively\footnote{It follows directly from Definition \ref{DefBound} that one cannot have one lying to the future and the other to the past.}. Let \(\gamma : (-\varepsilon, 0) \to U\) be now a (necessarily) future directed timelike curve with \(\lim_{s \nearrow 0} \gamma(s) = p\). 

We first show that\footnote{Although actually no confusion can arise, we write \(I^-(p,M)\) to emphasise that this denotes the past of \(p\) \emph{in} \(M\).} \(\psi\big(I^-(p,M) \cap U\big) = I^-(p', M') \cap \psi(U)\). 

So let \(q \in I^-(p,M) \cap U\). Then \(I^+(q,M)\) is an open neighbourhood of \(p\). Moreover, let \(t'_1 \in M'\) with \(t'_1 \gg p'\). Then \(I^-(t'_1,M')\) is an open neighbourhood of \(p'\). Since \(p\) and \(p'\) are corresponding boundary points, it follows that \( \psi^{-1}\big(I^-(t'_1,M')\cap \psi(U)\big) \cap I^+(q,M)   \neq \emptyset\). Thus we can find an \(r'_1 \in \psi(U)\) with \(\psi(q) \ll r'_1 \ll t'_1\); hence, in particular, \(\psi(q) \leq t'_1\). 
Taking a sequence \(t'_i \gg p'\), \(i \in \N\), with \(t'_i \to p'\) for \(i \to \infty\), we get \(\psi(q) \leq p'\) since the relation \(\leq\) is closed on globally hyperbolic manifolds\footnote{Cf.\ Lemma \(22\) in Chapter \(14\) of \cite{ONeill}.}.

In order to get \(\psi(q) \ll p'\), take an \(s \in U\) with \(q \ll s \ll p\) and repeat the argument above with \(s\) instead of \(q\). This then gives \(\psi(q) \ll \psi(s) \leq p'\), and thus\footnote{Cf.\ Proposition \(46\) in Chapter \(10\) of \cite{ONeill}.} \(\psi(q) \ll p'\). Hence, we have shown \(\psi\big(I^-(p,M) \cap U\big) \subseteq I^-(p', M') \cap \psi(U)\). The other inclusion follows by symmetry.

Let now \(\gamma : (-\varepsilon , 0) \to M\) be a future directed timelike curve with \(\lim_{s \nearrow 0} \gamma(s) = p\). Then \(\psi \circ \gamma|_{(-\varepsilon, 0)}\) is a timelike curve in \(I^-(p', M')\) and we claim that \(\lim_{t \nearrow 0} (\psi \circ \gamma) (t) = p'\). To see this, let \(V'\) be an open neighbourhood of \(p'\). Since \(M'\) satisfies the strong causality condition, we can find a \(q' \in V' \cap I^-(p',M')\) such that \(I^+(q',M') \cap I^-(p',M') \subseteq V'\).\footnote{Recall that the strong causality condition is satisfied at the point \(p'\) iff for all neighbourhoods \(V'\) of \(p'\) there is a neighbourhood \(W'\) of \(p'\) such that all causal curves with endpoints in \(W'\) are entirely contained in \(V'\). In order to prove the just made claim, it remains to pick a point \(q' \in W' \cap I^-(p',M')\).}

\begin{center}
\def\svgwidth{10cm}
\begingroup%
  \makeatletter%
  \providecommand\color[2][]{%
    \errmessage{(Inkscape) Color is used for the text in Inkscape, but the package 'color.sty' is not loaded}%
    \renewcommand\color[2][]{}%
  }%
  \providecommand\transparent[1]{%
    \errmessage{(Inkscape) Transparency is used (non-zero) for the text in Inkscape, but the package 'transparent.sty' is not loaded}%
    \renewcommand\transparent[1]{}%
  }%
  \providecommand\rotatebox[2]{#2}%
  \ifx\svgwidth\undefined%
    \setlength{\unitlength}{551.10341797bp}%
    \ifx\svgscale\undefined%
      \relax%
    \else%
      \setlength{\unitlength}{\unitlength * \real{\svgscale}}%
    \fi%
  \else%
    \setlength{\unitlength}{\svgwidth}%
  \fi%
  \global\let\svgwidth\undefined%
  \global\let\svgscale\undefined%
  \makeatother%
  \begin{picture}(1,0.41229115)%
    \put(0,0){\includegraphics[width=\unitlength]{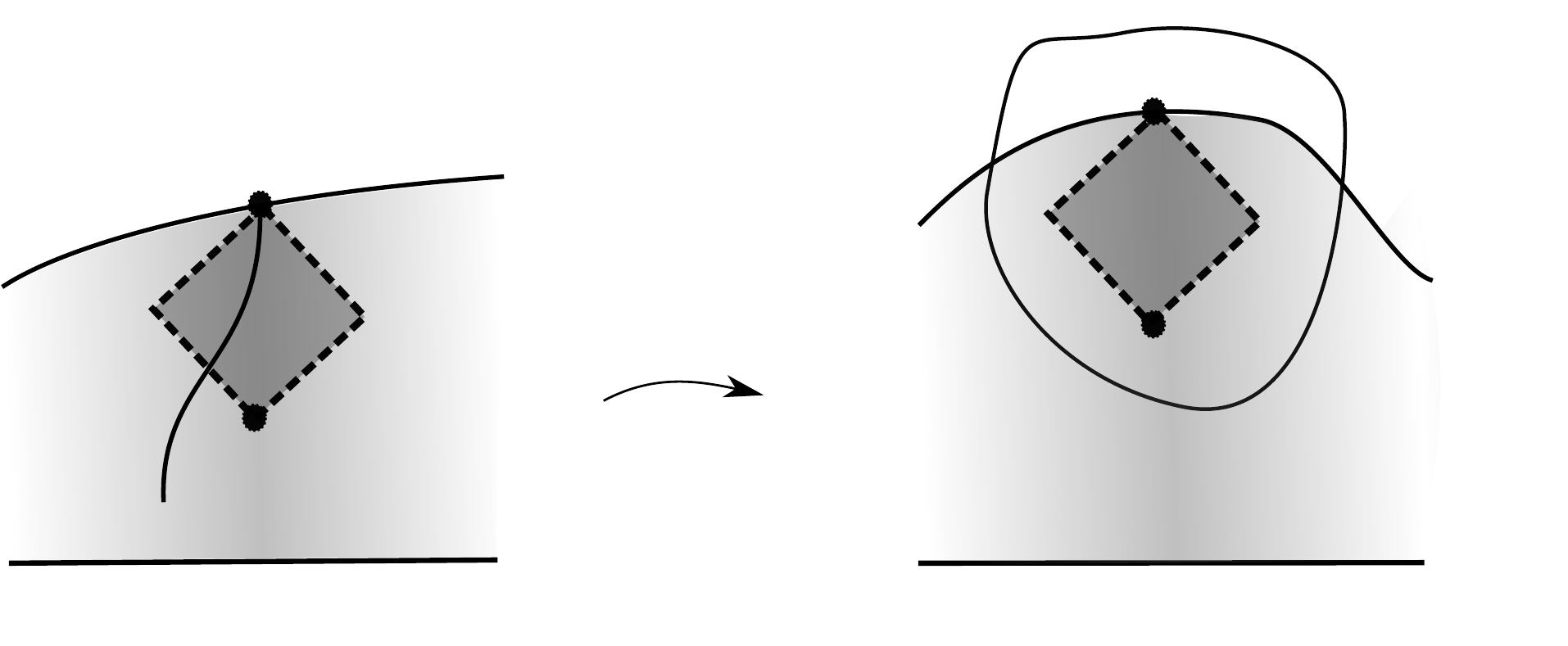}}%
    \put(0.41363025,0.18818018){\color[rgb]{0,0,0}\makebox(0,0)[lb]{\smash{\(\psi\)}}}%
    \put(0.02548104,0.28068365){\color[rgb]{0,0,0}\makebox(0,0)[lb]{\smash{\(M\)}}}%
    \put(0.58025335,0.32682745){\color[rgb]{0,0,0}\makebox(0,0)[lb]{\smash{\(M'\)}}}%
    \put(0.13444508,0.01252521){\color[rgb]{0,0,0}\makebox(0,0)[lb]{\smash{\(\iota(\overline{M})\)}}}%
    \put(0.71654686,0.01099782){\color[rgb]{0,0,0}\makebox(0,0)[lb]{\smash{\(\iota'(\overline{M})\)}}}%
    \put(0.75339676,0.18365282){\color[rgb]{0,0,0}\makebox(0,0)[lb]{\smash{\(q'\)}}}%
    \put(0.75096095,0.35371611){\color[rgb]{0,0,0}\makebox(0,0)[lb]{\smash{\(p'\)}}}%
    \put(0.17206882,0.29763823){\color[rgb]{0,0,0}\makebox(0,0)[lb]{\smash{\(p\)}}}%
    \put(0.17597221,0.1207432){\color[rgb]{0,0,0}\makebox(0,0)[lb]{\smash{\(q\)}}}%
    \put(0.07349761,0.13522684){\color[rgb]{0,0,0}\makebox(0,0)[lb]{\smash{ \(\gamma\)}}}%
    \put(0.86129464,0.36802405){\color[rgb]{0,0,0}\makebox(0,0)[lb]{\smash{\(V'\)}}}%
    \put(0.27412091,0.31797686){\color[rgb]{0,0,0}\makebox(0,0)[lb]{\smash{\(\partial U\)}}}%
    \put(0.89559585,0.26127368){\color[rgb]{0,0,0}\makebox(0,0)[lb]{\smash{\(\partial\psi(U)\)}}}%
  \end{picture}%
\endgroup%

\end{center}

From what we first showed, we know that \(q := \psi^{-1}(q') \in I^-(p,M)\). Since \(I^+(q,M)\) is an open neighbourhood of \(p\), there exists a \(\delta >0\) such that \(\gamma(s) \in I^+(q,M) \cap I^-(p,M)\) for all \(-\delta < s < 0\). Moreover, we have \(\psi \big(I^+(q,M) \cap I^-(p,M)\big) = I^+(q',M') \cap I^-(p',M')\), from which it follows that \((\psi \circ \gamma)(s) \in V'\) for all \(-\delta < s < 0\).
\end{proof}

If \(U\) is a CGHD of \(M\) and \(M'\) with isometric embedding \(\psi : U \subseteq M \to M'\), we denote the set of points in \(\partial U\) that have a corresponding boundary point in \(\partial \psi(U)\) with \(C\). 

\begin{lemma}
\label{ExtendIso}
Let \(U\) be a CGHD of \(M\) and \(M'\) with isometric embedding \(\psi : U \subseteq M \to M'\). Then the set \(C\) is open in \(\partial U\) and the isometric embedding \(\psi : U \to M'\) extends smoothly to \(\psi : U \cup C \to M'\).
\end{lemma}

\begin{proof}
Assume that there exists a pair \(p \in \partial U\) and \(p' \in \partial \psi(U)\) of corresponding boundary points, otherwise there is nothing to show. 

Let \(V \subseteq M\) be a convex\footnote{Recall that an open set is called convex iff it is a normal neighbourhood of each of its points. For the existence of convex neighbourhoods we refer the reader to Proposition \(7\) of Chapter \(5\) of \cite{ONeill}.} neighbourhood of \(p\) and \(V' \subseteq M'\) be a convex neighbourhood of \(p'\). Consider a future directed timelike geodesic \(\gamma : [-\varepsilon , 0) \to U\) with \(\lim_{s \nearrow 0} \gamma(s) = p\). Then, by Proposition \ref{EquivCha}, \(\gamma' := \psi \circ \gamma\) is a future directed timelike geodesic in \(M'\) with \(\lim_{s \nearrow 0} \gamma'(s) = p'\). Without loss of generality we may assume that \(\varepsilon >0\) is so small that \(\gamma([-\varepsilon,0)) \subseteq V\) and \(\gamma'([-\varepsilon,0)) \subseteq V'\). 

Let \(p \in W \subseteq V\) be a small open neighbourhood of \(p\) such that \(W \subseteq I^+(\gamma(-\varepsilon))\) and
\begin{equation*}
\psi_*\big[\exp^{-1}_{\gamma(-\varepsilon)}(W)\big] \subseteq \exp^{-1}_{\gamma'(-\varepsilon)}(V')\;.
\end{equation*}

We can now define the smooth extension \(\hat{\psi} : W \to M'\) by
\begin{equation*}
\hat{\psi} (q) := \exp_{\gamma'(-\varepsilon)}\big(\psi_*(\exp^{-1}_{\gamma(-\varepsilon)}(q))\big) \;.
\end{equation*}
This is clearly a smooth diffeomorphism onto its image and it also agrees with \(\psi\) on \(W \cap U\), since the exponential map commutes with isometries: Let \(q \in W\cap U\) and say \(X \in T_{\gamma(-\varepsilon)}M\) is such that \(q = \exp_{\gamma(-\varepsilon)}(X)\). We then have
\begin{equation*}
\psi(q)= \psi\big(\exp_{\gamma(-\varepsilon)}(X)\big) = \exp_{(\psi \circ \gamma) (-\varepsilon)} \big(\psi_*(X)\big) = \hat{\psi}(q)\;.
\end{equation*}
Moreover, using the same argument, we have \(W \cap \partial U \subseteq C\), since for \(q \in W \cap \partial U\) and \(X := \exp^{-1}_{\gamma(-\varepsilon)}(q)\), we have that \(s \mapsto \gamma(s) = \exp_{\gamma(-\varepsilon)}(s\cdot X)\) is a timelike curve that converges to \(q\) for \(s \nearrow 1\), while \((\psi \circ \gamma) (s)\) converges to a point in \(\partial \psi(U)\) for \(s \nearrow 1\). By Proposition \ref{EquivCha}, point \(iii)\), \(q\) thus has a corresponding boundary point.  Hence, \(C\) is open in \(\partial U\).
\end{proof}

Note that in the case of \(C\) being non-empty, this lemma states that one can extend the identification of \(M\) with \(M'\). It thus furnishes the closure part of the analogy to the method of continuity referred to in the introduction. Pursuing this analogy, the next two lemmata lay the foundation for restarting the local uniqueness argument again, i.e., they lay the foundation for the openness part.

\begin{lemma}
\label{SpacelikePoint}
Let \(U\) be a CGHD of \(M\) and \(M'\) with isometric embedding \(\psi : U \subseteq M \to M'\). Assume that 
\(C \cap J^+\big(\iota(\overline{M})\big)\) is non-empty. Then there exists a point \(p \in C\) with the property
\begin{equation}
\label{spacelike}
J^-(p) \cap \partial U \cap J^+\big(\iota(\overline{M})\big)  = \{p\}\;.
\end{equation}  
\end{lemma}

Whenever \(C\) is non-empty, we can assume without loss of generality (otherwise we reverse the time orientation) that we have in fact \(C \cap J^+\big(\iota(\overline{M})\big) \neq \emptyset\). In this case, the above lemma ensures the existence of a `spacelike' part of the boundary - only those parts are suitable for restarting the local uniqueness argument.

\begin{proof}
So assume that \(C \cap J^+\big(\iota(\overline{M})\big)\) is non-empty. Let \(p \in C \cap J^+\big(\iota(\overline{M})\big)\) and we have to deal with the case that \(\big(J^-(p) \cap \partial U \cap J^+\big(\iota(\overline{M})\big)\big) \setminus \{p\} \neq \emptyset\). So let \(q \in \big(J^-(p) \cap \partial U \cap J^+\big(\iota(\overline{M})\big) \big) \setminus \{p\}\). Thus, there exists a past directed causal curve \(\gamma\) from \(p\) to \(q\). Since \(\partial U \cap J^+\big(\iota(\overline{M})\big) \) is achronal, \(\gamma\) must be a null geodesic\footnote{\label{Ach}That \(\partial U \cap J^+\big(\iota(\overline{M})\big) \) is  achronal follows from \(\ll\) being an open relation, see Lemma \(3\) in Chapter \(14\) of \cite{ONeill}: If there were two points \(x, y \in \partial U \cap J^+\big(\iota(\overline{M})\big)\) with \(x \ll y\), then we could also find \(x' \in U^c \cap J^+\big(\iota(\overline{M})\big)\) close to \(x\) and \(y' \in U \cap J^+\big(\iota(\overline{M})\big)\) close to \(y\) such that \(x' \ll y'\). This, however, gives rise to an inextendible timelike curve in \(U\) which does not intersect the Cauchy hypersurface \(\iota(\overline{M})\) - a contradiction to the global hyperbolicity of \(U\). That \(\gamma\) must be a null geodesic is an easy consequence of the fundamental Proposition \(46\) in Chapter \(10\) of \cite{ONeill}.}. Let \(\gamma : [0, a) \to M\), where \(a >1\), be a parameterization of the past inextendible null geodesic \(\gamma\) with \(\gamma(0) = p\) and \(\gamma(1) = q\).  Moreover, note that \(\gamma([0,1]) \subseteq \partial U\). Since if there were a \(0<t<1\) with \(\gamma(t) \in U\) then global hyperbolicity of \(U\) would imply that \(\gamma(1) =q \in U\) as well. On the other hand, if \(\gamma(t) \in U^c \setminus \partial U\) then we could find a closeby point \(r \in U^c \setminus \partial U\) that could be connected by a timelike curve to \(p\). But then, we could also find a point \(s \in U\) close by to \(p\) such that \(r\) and \(s\) could be connected by a timelike curve - again a contradiction to the global hyperbolicity of \(U\).

Let \([0,b] := \gamma^{-1}(\partial U)\). Since \(\partial U\) is closed in \(M\), this is indeed a closed interval - and exactly the same argument as above shows that it is connected. In the following we show that \(\gamma(b)\) has the wanted property, namely
\begin{equation*}
\gamma(b) \in  C \textnormal{ and }  J^-(\gamma(b)) \cap \partial U \cap J^+\big(\iota(\overline{M})\big) = \{\gamma(b)\}\;.
\end{equation*} 
We first show that \(J:=\{t \in [0,b]\, |\, \gamma(t) \in C\}\) is equal to \([0,b]\).
Since \(\gamma(0) \in C\), \(J\) is non-empty. By Lemma \ref{ExtendIso} we know that \(C\) is open in \(\partial U\), so \(J\) is open in \([0,b]\). It remains to show that \(J\) is closed in \([0,b]\) in order to deduce that \(J = [0,b]\). 

Since by Lemma \ref{ExtendIso} \(\psi\) extends to an isometric embedding on \(U \cup C\), \(\gamma'|_J := \psi \circ \gamma|_J\) is a null geodesic in \(M'\). Denote with \(\gamma'\) the corresponding past inextedible null geodesic in \(M'\). So let \(t_j \in J\), \(j \in \N\), be a sequence with \(t_j \to t_\infty\) in \([0,b]\) for \(j \to \infty\). We then claim that \(\gamma'(t_\infty)\) and \(\gamma(t_\infty)\) are corresponding boundary points. This is seen as follows: let \(V \subseteq M\) be a neighbourhood of \(\gamma(t_\infty)\) and \(V' \subseteq M'\) a neighbourhood of \(\gamma'(t_\infty)\). Consider now a sequence of future directed timelike curves \(\alpha_j : (-\varepsilon, 0) \to U\), \(j \in \N\), with \(\lim_{s \nearrow 0}\alpha_j(s) = \gamma(t_j)\). Then for \(j\) large enough and \(\sigma <0\) close enough to zero, we have \(\alpha_j(\sigma) \in  V \cap \psi^{-1}\big(V' \cap \psi(U)\big)\). This finally shows that \(\gamma(b) \in C\). 

That \(\gamma(b)\) lies to the future of \(\iota(\overline{M})\) is immediate, since \(\gamma\) cannot cross \(\iota(\overline{M})\) as long as it lies in \(\partial U\).

In order to show that \(J^-(\gamma(b)) \cap \partial U \cap J^+\big(\iota(\overline{M})\big) = \{\gamma(b)\}\), assume that there were a \(q \in \big(J^-(\gamma(b)) \cap \partial U \cap J^+\big(\iota(\overline{M})\big)\big) \setminus \{\gamma(b)\}\). Then there is a past directed null geodesic from \(\gamma(b)\) to \(q\). Concatenate \(\gamma|_{[0,b]}\) and this null geodesic. Note that by definition of \([0,b]\) this null geodesic must be broken. But then we can connect \(p\) and \(q\) by a timelike curve\footnote{See again Proposition \(46\) in Chapter \(10\) of \cite{ONeill}.}, which, as before, leads to a contradiction to \(U\) being globally hyperbolic.
\end{proof}

\begin{lemma}
\label{FuturePoint}
Let \(U\) be a GHD of some initial data and \(M \supseteq U\) an extension of \(U\). Suppose that there exists a \(p \in \partial U\) that satisfies \eqref{spacelike}. Then for every open neighbourhood \(W\) of \(p\) in \(M\) there exists a point \(q \in I^+(p) \subseteq M\) such that 
\begin{equation*}
J^-(q) \cap U^c \cap J^+\big(\iota(\overline{M})\big)  \subseteq W\;.
\end{equation*}
\end{lemma}

\begin{proof}
So let \(p\) satisfy \(J^-(p) \cap \partial U \cap J^+\big(\iota(\overline{M})\big) = \{p\}\). Let \(\gamma : [0, \varepsilon] \to M\) be a future directed timelike curve with \(\gamma(0) = p\). Then we have \(\gamma((0,\varepsilon]) \subseteq U^c\). Let \(W \subseteq M\) be an open neighbourhood of \(p\).  If the lemma were not true, then there is a sequence \(t_j \in (0, \varepsilon]\), \(j\in \N\), with \(t_j \to 0\) in \([0,\varepsilon]\) for \(j \to \infty\), and a sequence of points \(\{q_j\}_{j\in\N}\) with 
\begin{equation*}
q_j \in J^-(\gamma(t_j)) \cap U^c \cap J^+\big(\iota(\overline{M})\big) \cap W^c \;.
\end{equation*}
Since \(M\) is globally hyperbolic, \(J^-(\gamma(\varepsilon)) \cap J^+\big(\iota(\overline{M})\big)\) is compact, thus \(J^-(\gamma(\varepsilon)) \cap U^c \cap J^+\big(\iota(\overline{M})\big) \cap W^c\) is compact, and we can assume without loss of generality that \(q_j \to q \in J^-(\gamma(\varepsilon)) \cap U^c \cap J^+\big(\iota(\overline{M})\big) \cap W^c\). Since the relation \(\leq\) is closed, we obtain \(q \leq p\), and thus clearly \(q < p\).  But this leads again to a contradiction: We cannot have \(q \in \partial U\) by assumption, thus \(q \in U^c \setminus \partial U\). This, however, contradicts the global hyperbolicity of \(U\) in the same way as we argued in the proof of Lemma \ref{SpacelikePoint}. 
\end{proof}

We are finally well-prepared for the proof of Theorem \ref{NotMCGHD}.

\begin{proof}[Proof of Theorem \ref{NotMCGHD}:]
Recall that \(M\) and \(M'\) are GHDs, and \(U \subseteq M\) is a CGHD of \(M\) and \(M'\) that has corresponding boundary points in \(M\) and \(M'\). Without loss of generality we can assume that \(C \cap J^+\big(\iota(\overline{M})\big)\) is non empty, and thus, by Lemma \ref{SpacelikePoint}, we can find a \(p \in C\) which satisfies \(J^-(p) \cap \partial U \cap J^+\big(\iota(\overline{M})\big) = \{p\}\). Since by Lemma \ref{ExtendIso} \(C\) is open in \(\partial U\), we can find a convex neighbourhood \(V \subseteq M\) of \(p\) such that \(V \cap \partial U \subseteq C\). Since the strong causality condition holds at \(p\), we can find a causally convex neighbourhood \(W\) of \(p\) whose closure is compact and completely contained in  \(V\).\footnote{Recall that an open set \(W \subseteq M\) is called \emph{causally convex} iff every causal curve in \(M\) with endpoints in \(W\) is entirely contained in \(W\). That we can find such a causally convex neighbourhood follows from the strong causality condition: Let \(V_1\) be a neighbourhood of \(p\) whose closure is compact and completely contained in \(V\). By the strong causality condition we can find a neighbourhood \(V_2 \subseteq V_1\) of \(p\) with the property that every causal curve with endpoints in \(V_2\) is completely contained in \(V_1\). Pick now two points \(p_1, p_2 \in V_2\) such that  \(p_1 \ll p \ll p_2\). It follows that \(W:=I^+(p_1) \cap I^-(p_2)\) is an open neighbourhood of \(p\) which is completely contained in \(V_1\) and thus has compact closure. Moreover, \(W\) is causally convex: Let \(\gamma\) be a causal curve with endpoints \(x \leq y \in W\) and let \(z\) be a point on \(\gamma\). We then have \(p_1 \ll x \leq z \leq y \ll p_2\), and by Proposition \(46\) of Chapter \(10\) in \cite{ONeill} it follows that \(z \in W\).}
Let \(q \in I^+(p)\) be a point with the property that \(J^-(q) \cap U^c \cap J^+\big(\iota(\overline{M})\big) \subseteq W\), whose existence is guaranteed by Lemma \ref{FuturePoint}.

Let us denote with  \(\tau_q : M \to [0,\infty)\) the \emph{time separation} from \(q\), i.e.\
\begin{quote}
\(\tau_q(r) := \sup \{L(\gamma) : \gamma \) is a future directed causal curve segment from \(r\) to \(q\)\},
\end{quote}
where \(L(\gamma)\) denotes the length of \(\gamma\). If \(r \notin J^-(q)\) we set \(\tau_q(r)\) equal to zero. Note that \(\tau_q\) restricted to \(W\) can be explicitly given by the exponential map based at \(q\): Given \(r \in W\), there exists, by the global hyperbolicity of \(M\), a geodesic from \(r\) to \(q\) whose length equals the time separation from \(r\) to \(q\). Since \(W\) is causally convex, this geodesic must be completely contained in \(W\) - and since \(V \supseteq W\) is convex, this geodesic is a radial one in the exponential chart centred at \(q\). Thus, we obtain for $r \in I^-(q) \cap W$ that 
\begin{equation}
\label{GivenByExp}
\tau_q(r) = \sqrt{-g|_q\big(\exp_q^{-1}(r), \exp_q^{-1}(r)\big)}\;.
\end{equation}
In particular \(\tau_q\) is smooth in \(I^-(q) \cap W\) and, by the global hyperbolicity of \(M\), continuous in \(V\).\footnote{Cf.\ Lemma \(21\) in Chapter \(14\) of \cite{ONeill}.}

\begin{center}
\def\svgwidth{5cm}
\begingroup%
  \makeatletter%
  \providecommand\color[2][]{%
    \errmessage{(Inkscape) Color is used for the text in Inkscape, but the package 'color.sty' is not loaded}%
    \renewcommand\color[2][]{}%
  }%
  \providecommand\transparent[1]{%
    \errmessage{(Inkscape) Transparency is used (non-zero) for the text in Inkscape, but the package 'transparent.sty' is not loaded}%
    \renewcommand\transparent[1]{}%
  }%
  \providecommand\rotatebox[2]{#2}%
  \ifx\svgwidth\undefined%
    \setlength{\unitlength}{280.57802734bp}%
    \ifx\svgscale\undefined%
      \relax%
    \else%
      \setlength{\unitlength}{\unitlength * \real{\svgscale}}%
    \fi%
  \else%
    \setlength{\unitlength}{\svgwidth}%
  \fi%
  \global\let\svgwidth\undefined%
  \global\let\svgscale\undefined%
  \makeatother%
  \begin{picture}(1,1.00267114)%
    \put(0,0){\includegraphics[width=\unitlength]{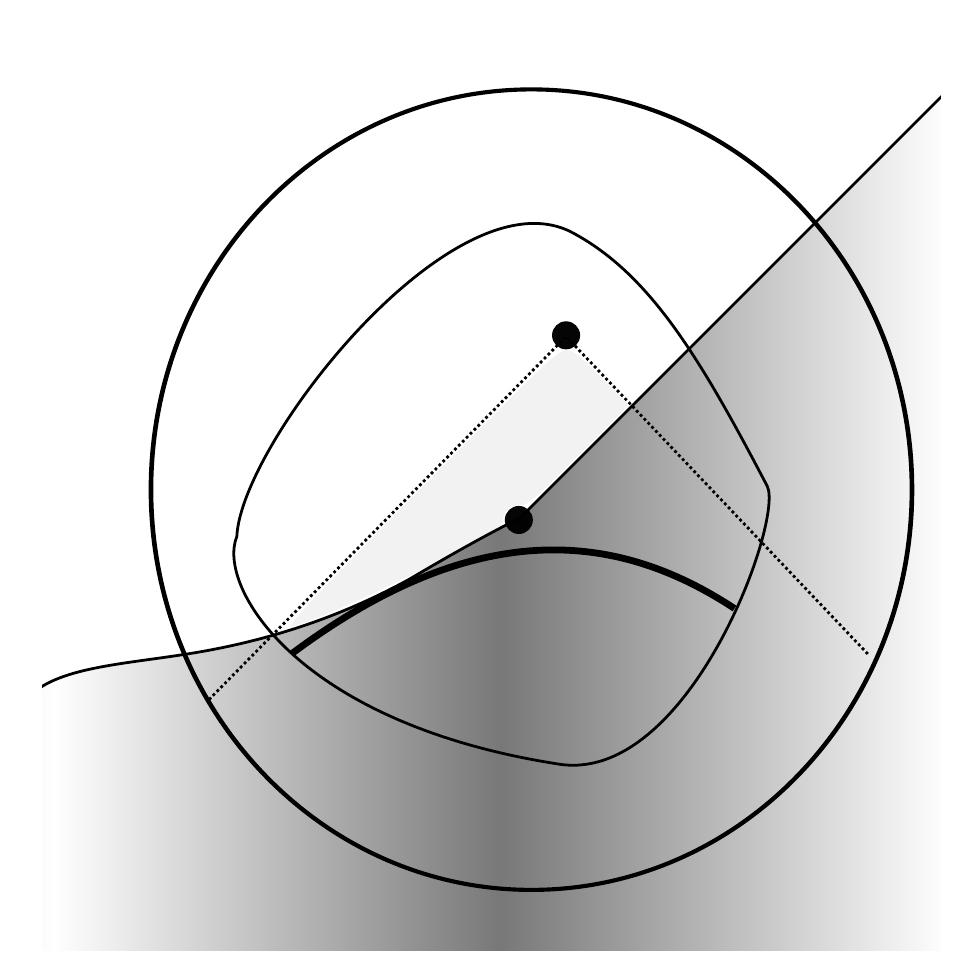}}%
    \put(0.57281632,0.44750099){\color[rgb]{0,0,0}\makebox(0,0)[lb]{\smash{\(p\)}}}%
    \put(0.52543132,0.67234445){\color[rgb]{0,0,0}\makebox(0,0)[lb]{\smash{\(q\)}}}%
    \put(0.6347096,0.32326502){\color[rgb]{0,0,0}\makebox(0,0)[lb]{\smash{\(S\)}}}%
    \put(0.06280003,0.12535223){\color[rgb]{0,0,0}\makebox(0,0)[lb]{\smash{\(U\)}}}%
    \put(0.7166064,0.90300736){\color[rgb]{0,0,0}\makebox(0,0)[lb]{\smash{\(V\)}}}%
    \put(0.73635854,0.59651275){\color[rgb]{0,0,0}\makebox(0,0)[lb]{\smash{\(W\)}}}%
  \end{picture}%
\endgroup%

\end{center}

Since \(\overline{W}\) is compact, \(\tau_q\) takes on its maximum on \(\overline{W} \cap U^c \cap J^+\big(\iota(\overline{M})\big)\). Let us denote this maximum by \(\tau_0\). Clearly, we have \(\tau_0 > 0\). Moreover, one has \(\tau_q (r) = \tau_0\) only for \(r \in \partial U \cap W \cap J^+\big(\iota(\overline{M})\big)\), since if this were not the case, using normal coodinates around \(q\), one could continue the length maximising geodesic from \(r_0\) to \(q\) a bit to the past, staying in \(W \cap U^c\), which would lead to a longer timelike curve.

We now define 
\begin{equation*}
S:= \tau_q^{-1} (\tau_0) \cap W \cap I^+\big(\iota(\overline{M})\big) \;.
\end{equation*}
By construction, $S$ contains at least one point  of \(\partial U\); and since the hyperboloid 
\begin{equation*}
Q^-_{\tau_0} := \big{\{}X \in T_qM \; \big| \; \sqrt{-g|_q(X,X)} = \tau_0 \text{ and } X \text{ past directed}\big{\}}
\end{equation*}
is smooth, \eqref{GivenByExp} shows that $S$ is smooth as well. Moreover, it follows from the Gauss lemma\footnote{Cf.\ Lemma 1 in Chapter 5 of \cite{ONeill}} that the normal of $S$ at $\exp_q(X)$, where $X \in Q^-_{\tau_0}$, is given by $(\exp_q)_*(X)$, which is timelike - and hence $S$ is spacelike.
Furthermore, \(S\) is contained in \(\overline{U} \cap J^+\big(\iota(\overline{M})\big)\), since \(\tau_q(r)\) is only greater than zero for \(r \in J^-(q)\), and on \(J^-(q) \cap U^c \cap J^+\big(\iota(\overline{M})\big) \subseteq W\) we only have \(\tau_q(r) = \tau_0\) for \(r \in \partial U\) as argued above. 

Using Lemma \ref{ExtendIso} (and therefore the fact that \(V \cap \partial U \subseteq C\)) we can thus map\footnote{Recall that we denote the isometric embedding of \(U\) into \(M'\) by \(\psi\).} \(S\) isometrically to \(\psi(S) \subseteq M'\) - and suitable neighbourhoods of \(S\) in \(M\) and of \(\psi(S)\) in \(M'\) are GHDs of \((S,\bar{g}_S, k_S)\) (where \(\bar{g}_S\) is the induced metric from the ambient spacetime \(M\) and \(k_S\) is the second fundamental form of \(S\) in \(M\)). By Theorem \ref{LocalTheory} there exists a globally hyperbolic development  \(N \subseteq M\) of \((S, \bar{g}_S,k_S)\) together with an isometric embedding \( \phi : N \to M'\) such that \(\phi|_S = \psi|_S\). 

We now claim that \(\psi = \phi\) holds in \(N \cap \overline{U}\), which would imply that we can extend \(\psi\) to an isometric embedding \(\Psi : U \cup N  \to M'\). By the same argument as in the proof of  Corollary \ref{welldefined} we obtain $(d\psi)|_S = (d\phi)|_S$. The same continuity argument as in the proof of Lemma \ref{IsoConnected}, but this time applied to \(N \cap \overline{U}\), now proves the claim.

Also note that  \(U \cup N\) is globally hyperbolic with Cauchy hypersurface \(\iota(\overline{M})\): consider a point $r$ on an inextendible timelike curve $\gamma$ in  $U \cup N$. If $r$ is in $N \setminus U$, the curve $\gamma$ must intersect $S$, since $S$ is a Cauchy hypersurface in $N$. The choice of $\tau_0$ then implies that $\gamma$ must also enter $U$. So without loss of generality we can assume that there exists a point $r$ on $\gamma$ that lies in $U$. But since $U$ is globally hyperbolic with Cauchy hypersurface \(\iota(\overline{M})\), it now follows that $\gamma$ must intersect \(\iota(\overline{M})\). Moreover, $\gamma$ cannot intersect \(\iota(\overline{M})\) more than once, since \(\iota(\overline{M})\) is also a Cauchy hypersurface for $M$. 

Finally, since \(S\) contains at least one point in \(\partial U\), it follows that \(U \cup N \subseteq M\) is a strictly larger CGHD of \(M\) and \(M'\) than the CGHD \(U\) we started with. 
\end{proof}

Invoking the tertium non datur, Theorem \ref{NotMCGHD} implies
\begin{theorem}
\label{TNDVersion}
Let \(M\) and \(M'\) be GHDs of the same initial data, and let \(U\) be the MCGHD of \(M\) and \(M'\). Then \(U\) does not have corresponding boundary points in \(M\) and \(M'\).
\end{theorem}

\subsection{Finishing off the proof of the main theorems}
\label{Final}

From here on, the proof of Theorem \ref{CommonExtension} is straightforward:

\begin{proof}[Proof of Theorem \ref{CommonExtension}:]
As already outlined in the introduction, we will construct the common extension of \(M\) and \(M'\) by glueing them together along their MCGHD. Theorem \ref{TNDVersion} will yield that this space is Hausdorff. It then remains to show that this quotient space comes with enough natural structure that turns it into a GHD.

Thus, let us take the disjoint union \(M \sqcup M'\) of \(M\) and \(M'\) and endow it with the natural topology. Let us denote the MCGHD of \(M\) and \(M'\) by \(U\) (the existence of such a CGHD is guaranteed by Theorem \ref{ExMCGHD}) and the isometric embedding of \(U\) into \(M'\) by \(\psi\). We now consider the following equivalence relation on \(M \sqcup M'\): For \(p,q \in M\sqcup M'\) we define \(p \sim q\) if and only if
\begin{quote}
\(p \in U\subseteq M\) and \(q = \psi(p)\) \hspace{1cm} \emph{or} \hspace{1cm} \(q \in U \subseteq M\) and \(p = \psi(q)\) \hspace{1cm} \emph{or} \hspace{1cm} \(p=q\).
\end{quote}
We then take the quotient \((M \sqcup M')/_\sim =:\tilde{M}\), endowed with the quotient topology. 
The following elementary remark is needed in the remainder of the proof:
\begin{equation}
\label{homeo}
\textnormal{The maps } \pi \circ j \textnormal{ and } \pi \circ j' \textnormal{ are homeomorphisms onto their image.}
\end{equation}

\begin{center}
\def\svgwidth{8cm}
\begingroup%
  \makeatletter%
  \providecommand\color[2][]{%
    \errmessage{(Inkscape) Color is used for the text in Inkscape, but the package 'color.sty' is not loaded}%
    \renewcommand\color[2][]{}%
  }%
  \providecommand\transparent[1]{%
    \errmessage{(Inkscape) Transparency is used (non-zero) for the text in Inkscape, but the package 'transparent.sty' is not loaded}%
    \renewcommand\transparent[1]{}%
  }%
  \providecommand\rotatebox[2]{#2}%
  \ifx\svgwidth\undefined%
    \setlength{\unitlength}{627.48212891bp}%
    \ifx\svgscale\undefined%
      \relax%
    \else%
      \setlength{\unitlength}{\unitlength * \real{\svgscale}}%
    \fi%
  \else%
    \setlength{\unitlength}{\svgwidth}%
  \fi%
  \global\let\svgwidth\undefined%
  \global\let\svgscale\undefined%
  \makeatother%
  \begin{picture}(1,0.18878102)%
    \put(0,0){\includegraphics[width=\unitlength]{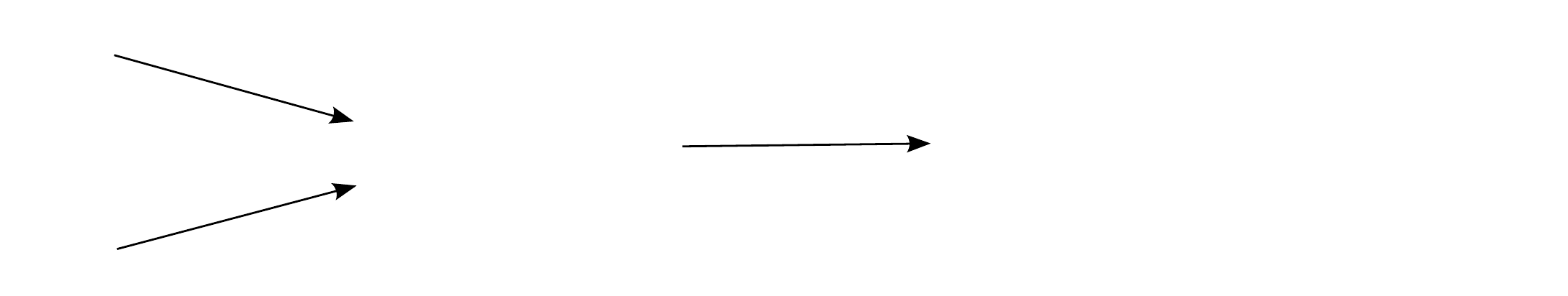}}%
    \put(0.00546403,0.14443821){\color[rgb]{0,0,0}\makebox(0,0)[lb]{\smash{\(M\)}}}%
    \put(-0,0.0169445){\color[rgb]{0,0,0}\makebox(0,0)[lb]{\smash{\(M'\)}}}%
    \put(0.13660039,0.15536622){\color[rgb]{0,0,0}\makebox(0,0)[lb]{\smash{\(j\)}}}%
    \put(0.13842173,0.00965914){\color[rgb]{0,0,0}\makebox(0,0)[lb]{\smash{\(j'\)}}}%
    \put(0.24588069,0.08433406){\color[rgb]{0,0,0}\makebox(0,0)[lb]{\smash{\(M \sqcup M'\)}}}%
    \put(0.48811864,0.1134754){\color[rgb]{0,0,0}\makebox(0,0)[lb]{\smash{\(\pi\)}}}%
    \put(0.61014827,0.08433402){\color[rgb]{0,0,0}\makebox(0,0)[lb]{\smash{\((M \sqcup M')/_\sim\)}}}%
  \end{picture}%
\endgroup%

\end{center}

Here the maps \(j\) and \(j'\) denote the canonical inclusion maps. Verifying  \eqref{homeo} is an easy exercise in set topology: Clearly the maps are continuous and injective. We show that they are also open: for \(A \subseteq M\) open we have, with slight abuse of notation, that \(M \cap \big[ \pi^{-1} \big( (\pi \circ j) (A) \big) \big] = A\) is open and so is
\(M' \cap \big[ \pi^{-1} \big((\pi \circ j) (A) \big) \big] = \psi(U \cap A)\).
\newline
\newline
We now show that the quotient topology on \(\tilde{M}\) is indeed Hausdorff. Using \eqref{homeo}, we can easily separate two points \([p] \neq [q] \in \tilde{M}\), if
\begin{enumerate}
\item \(p \neq q \in M\): In this case we separate \(p\) and \(q\) in \(M\) and then use the fact that \(\pi \circ j\) is a homeomorphism in order to push forward the separating neighbourhoods to \(\tilde{M}\).
\item \(p \in M \setminus \overline{U}\) and \(q \in M' \setminus \psi(U)\): we choose a neighbourhood of \(p\) in \(M\) that lies entirely in \(M \setminus \overline{U}\) and an arbitrary neighbourhood of \(q\) in \(M'\). Pushing forward these neighbourhoods via the homeomorphisms, we obtain separating neighbourhoods in \(\tilde{M}\).
\end{enumerate}
Trivial permutations or modifications of these two possibilities leave only open the task to separate \([p]\) and \([q]\) if \(p \in \partial U\) and \(q \in \partial \psi(U)\), or \(q \in \partial U\) and \(p \in \partial \psi(U)\). So suppose we could not separate these two points in, without loss of generality, the case \(p \in \partial U\) and \(q \in \partial \psi(U)\). For all neighbourhoods \(V\) of \(p\) and \(V'\) of \(p'\), we then have \((\pi \circ j)(V) \cap (\pi \circ j')(V') \neq \emptyset\). This, however, implies that \(\psi^{-1}\big( V' \cap \psi(U)\big) \cap V \neq \emptyset\), i.e., \(p\) and \(q\) are corresponding boundary points of \(U\) - in contradiction to Theorem \ref{TNDVersion}. Thus, \(\tilde{M}\) is indeed Hausdorff.
\newline
\newline
In the remaining part of the proof we show that \(\tilde{M}\) possesses a natural structure that turns it into a common extension of \(M\) and \(M'\).

\begin{enumerate}
\item \emph{\(\tilde{M}\) is locally euclidean and has a natural smooth structure:} We have to give an atlas for \(\tilde{M}\). Let \(\{V_i, \varphi_i\}_{i \in \N}\) be an atlas for \(M\) and \(\{V'_k,\varphi'_k\}_{k \in \N}\) an atlas for \(M'\), where the \(\varphi's\) are here homeomorphisms from some open subset of \(\R^{d+1}\) to the \(V's\). We then define an atlas for \(\tilde{M}\) by
\begin{equation*}
\Big{\{} (\pi \circ j) (V_i), \pi \circ j \circ \varphi_i \Big{\}}_{i \in \N} \cup \Big{\{} (\pi \circ j') (V_k), \pi \circ j' \circ \varphi_k \Big{\}}_{k \in \N}\;.
\end{equation*}
By \eqref{homeo} this furnishes an open covering of \(\tilde{M}\) and it is easy to check that the transition functions are either of the form \(\varphi_{i_0}^{-1} \circ \varphi_{i_1}\) with \(i_0, i_1 \in \N\), the primed analogue, or \((\varphi_{k_0}')^{-1} \circ \psi \circ \varphi_{i_0} \) with \(i_0, k_0 \in \N\), which are all smooth diffeomorphisms.

\item \emph{\(\tilde{M}\) is second countable:} This follows directly from the previous construction.

\item \emph{\(\tilde{M}\) has a natural smooth Lorentzian metric that is Ricci-flat:} Since \(\pi \circ j\) and \(\pi \circ j'\) are smooth diffeomorphism onto their image, we can endow \(\tilde{M}\) with a smooth Lorentzian metric by pushing forward \(g\) and \(g'\). On \((\pi \circ j)(U)\) the two metrics obtained in this way agree since \(\psi\) is an isometry, thus this yields a smooth Lorentzian Ricci-flat metric \(\tilde{g}\) on \(\tilde{M}\). Moreover, note that this turns \(\pi \circ j\) and \(\pi \circ j'\) into isometries.

\item \emph{\((\tilde{M}, \tilde{g})\) is globally hyperbolic with Cauchy surface \(\tilde{\iota}(\overline{M})\):} Here we have defined \(\tilde{\iota} := \pi \circ j \circ \iota : \overline{M} \to \tilde{M}\). So let \(\gamma : I \to \tilde{M}\) be an inextendible timelike curve, where \(I \subseteq \R\). Take \(t_0 \in I\) and, without loss of generality, assume \(\gamma(t_0) \in (\pi \circ j)(M)\). If we denote with \(J \ni t_0\) the maximal connected subinterval of \(I\) such that \(\gamma(J) \subseteq (\pi \circ j)(M)\), then \(\gamma|_J\) can be considered as an inextendible timelike curve in \(M\) and thus has to intersect \(\iota(\overline{M})\). Hence, \(\gamma\) intersects \(\tilde{\iota}(\overline{M})\) at least once. 

Let us now assume that \(\gamma\) intersected \(\tilde{\iota}(\overline{M})\) more than once. We can find \(t_1 < t_3 \in I\) with \(\gamma(t_1), \gamma(t_3) \in \tilde{\iota}(\overline{M})\) and \(\gamma(t) \notin \tilde{\iota}(\overline{M})\) for \(t_1 < t <t_3\). Since \(M\) and \(M'\) are globally hyperbolic, \(\gamma|_{[t_1,t_3]}\) cannot be contained entirely in \(\pi \circ j(M)\) or \(\pi \circ j'(M')\). Thus, there must be \(t_2, t_{12}, t_{23}\) with \(t_1 < t_{12} < t_2 < t_{23} < t_3\) such that \(\gamma(t_2) \in (\pi \circ j) (U)\) and, without loss of generality, \(\gamma(t_{12}) \notin (\pi \circ j')(M')\) and \(\gamma(t_{23}) \notin (\pi \circ j)(M)\).\footnote{The other possibility is \(\gamma(t_{12}) \notin (\pi \circ j)(M)\) and \(\gamma(t_{23}) \notin (\pi \circ j')(M')\) and leads in the same way to a contradiction.} But this leads to an inextendible timelike curve in \(U\) that does not intersect \(\iota(\overline{M})\), a contradiction, since \(U\) is globally hyperbolic.

\item \emph{\((\tilde{M}, \tilde{g})\) has a natural time orientation:} Since \(M\) and \(M'\) are time oriented, there exist continuous timelike vector fields \(T\) on \(M\) and \(T'\) on \(M'\). Since \(\psi : U \to M'\) preserves the time orientation, at each point \(\psi_* (T|_U)\) and \(T'|_{\psi(U)}\) lie in the same component of the set of all timelike tangent vectors at this point. Thus, pushing forward \(T\) and \(T'\) via \(\pi \circ j\) and \(\pi \circ j'\) we can consistently single out a future direction at each point of \(\tilde{M}\).
It remains to show that this choice is continuous. But since this is a local property, this follows immediately form \((\pi \circ j)_*(T)\) and \((\pi \circ j')_*(T')\) being continuous.
\end{enumerate}
We have thus shown that \((\tilde{M}, \tilde{g}, \tilde{\iota})\) is a GHD of \((\overline{M}, \bar{g}, \bar{k})\) and, moreover, it is an extension of \(M\) and \(M'\), where the isometric embeddings are given by the maps \(\pi \circ j\) and \(\pi \circ j'\). This finishes the proof of Theorem \ref{CommonExtension}.
\end{proof}

\label{detailed}
As outlined in the introduction, we would like to construct now the MGHD by glueing all GHDs together along their MCGHDs. However, the following subtlety arises: the collection of all GHDs of given initial data is not a set, but a \emph{proper} class - and thus we cannot use the axioms of the Zermelo-Fraenkel set theory for justifying the glueing construction we have in mind. Fortunately, there is an easy way to circumvent this obstacle: Instead of considering \emph{all} GHDs of given initial data \((\overline{M}, \bar{g}, \bar{k})\), we only consider those whose underlying manifold is a subset of \(\overline{M} \times \R\).\footnote{We will in fact impose some further restrictions on the GHDs, which are, however, not strictly necessary.} This collection \(X\) of GHDs is indeed a set (as we will show below), and thus we can glue all such GHDs together along their MCGHDs. In order to justify that the so obtained GHD \(\tilde{M}\) is indeed the MGHD, we just note that any GHD of the same initial data is isometric to one in \(X\), and hence isometrically embeds into \(\tilde{M}\).

\begin{proof}[Proof of Theorem \ref{MGHD}:]
We consider fixed initial data \((\overline{M}, \bar{g}, \bar{k})\). In the following we argue that \emph{the collection \(X\) of all GHDs \(M\) whose underlying manifold is an open neighbourhood of \(\overline{M} \times \{0\}\) in \(\overline{M} \times \R\) and whose embeddings \(\iota : \overline{M} \to M\) of the initial data into \(M\) are given by \(\iota(x) = (x,0)\), where \(x \in \overline{M}\), is a set}. 

To see this, consider the set \(Y:= T^*(\overline{M} \times \R) \otimes T^*(\overline{M} \times \R)\), i.e., the tensor product of the cotangent bundle of \(\overline{M} \times \R\) with itself. Each of the members of \(X\) is given by a subset of \(Y\). The \emph{axiom of power set} ensures that there is a set \(\mathcal{P}(Y)\) containing all subsets of \(Y\). The \emph{axiom schema of specification} now ensures that 
\begin{align*}
X:= \big{\{} M \in \mathcal{P}(Y) \; \big| \; &\overline{M} \times \{0\} \subseteq M \subseteq \overline{M} \times \R \textnormal{ is a GHD of the given initial data} \\ &\textnormal{and the initial data embeds canonically into } \overline{M} \times \{0\} \subseteq M \big{\}}
\end{align*}
is a set.

To simplify notation, let us now write \(X = \{M_\alpha \;|\; \alpha \in A\}\). We denote the MCGHD of \(M_{\alpha_i}\) and \(M_{\alpha_k}\) with \(U_{\alpha_i \alpha_k} \subseteq M_{\alpha_i}\) and the corresponding isometric embedding with \(\psi_{\alpha_i \alpha_k} : U_{\alpha_i \alpha_k} \to M_{\alpha_k}\). We define an equivalence relation \(\sim\) on \(\bigsqcup_{\alpha \in A} M_\alpha\) by
\begin{equation}
\label{equiv}
M_{\alpha_i} \ni p_{\alpha_i} \sim q_{\alpha_k} \in M_{\alpha_k} \textnormal{ iff } p_{\alpha_i} \in U_{\alpha_i \alpha_k} \textnormal{ and } \psi_{\alpha_i \alpha_k}(p_{\alpha_i}) = q_{\alpha_k}\;
\end{equation}
and take the quotient \((\bigsqcup_{\alpha \in A} M_\alpha)/_\sim =: \tilde{M}\) with the quotient topology. Note that \eqref{equiv} is indeed an equivalence relation. For the transitivity observe that if \(p_{\alpha_i} \in M_{\alpha_i}\), \(p_{\alpha_k} \in M_{\alpha_k}\) and \(p_{\alpha_l} \in M_{\alpha_l}\) with \(p_{\alpha_i} \sim p_{\alpha_k}\) and \(p_{\alpha_k} \sim p_{\alpha_l}\), then we have that \(U_{\alpha_i \alpha_k} \cap \psi_{\alpha_i \alpha_k}^{-1}(U_{\alpha_k \alpha_l})\) together with the composition \(\psi_{\alpha_k \alpha_l} \circ \psi_{\alpha_i \alpha_k}\) is a CGHD of \(M_{\alpha_i}\) and \(M_{\alpha_l}\) that contains \(p_{\alpha_i}\) and identifies it with \(p_{\alpha_l}\) - so certainly the MCGHD of \(M_{\alpha_i}\) and \(M_{\alpha_l}\) leads to the same identification.

\begin{enumerate}
\item \emph{\(\tilde{M}\) is Hausdorff:} Let \([p_{\alpha_i}] \neq [q_{\alpha_k}] \in \tilde{M}\) with \(p_{\alpha_i} \in M_{\alpha_i}\) and \(q_{\alpha_k} \in M_{\alpha_k}\). We show that we can find open neighbourhoods in \(\tilde{M}\) that separate these points. 

\begin{center}
\def\svgwidth{11cm}
\begingroup%
  \makeatletter%
  \providecommand\color[2][]{%
    \errmessage{(Inkscape) Color is used for the text in Inkscape, but the package 'color.sty' is not loaded}%
    \renewcommand\color[2][]{}%
  }%
  \providecommand\transparent[1]{%
    \errmessage{(Inkscape) Transparency is used (non-zero) for the text in Inkscape, but the package 'transparent.sty' is not loaded}%
    \renewcommand\transparent[1]{}%
  }%
  \providecommand\rotatebox[2]{#2}%
  \ifx\svgwidth\undefined%
    \setlength{\unitlength}{594.21430664bp}%
    \ifx\svgscale\undefined%
      \relax%
    \else%
      \setlength{\unitlength}{\unitlength * \real{\svgscale}}%
    \fi%
  \else%
    \setlength{\unitlength}{\svgwidth}%
  \fi%
  \global\let\svgwidth\undefined%
  \global\let\svgscale\undefined%
  \makeatother%
  \begin{picture}(1,0.29399341)%
    \put(0,0){\includegraphics[width=\unitlength]{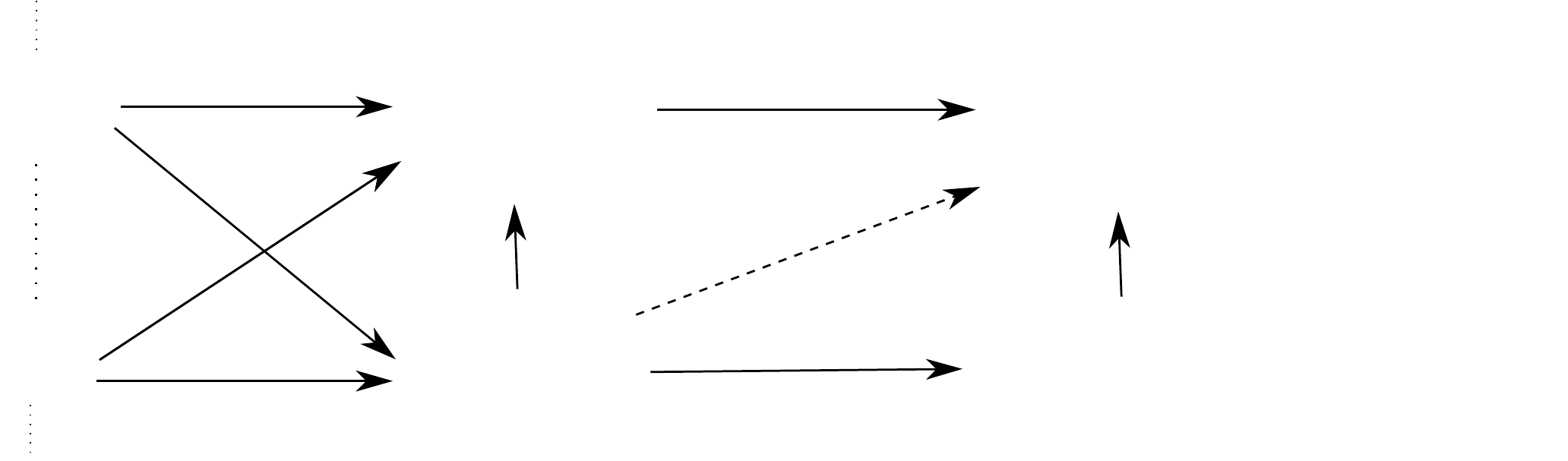}}%
    \put(0.00769323,0.21431433){\color[rgb]{0,0,0}\makebox(0,0)[lb]{\smash{\(M_{\alpha_i}\)}}}%
    \put(0,0.05660312){\color[rgb]{0,0,0}\makebox(0,0)[lb]{\smash{\(M_{\alpha_k}\)}}}%
    \put(0.12886164,0.1835414){\color[rgb]{0,0,0}\makebox(0,0)[lb]{\smash{\(j_i\)}}}%
    \put(0.12693833,0.07198952){\color[rgb]{0,0,0}\makebox(0,0)[lb]{\smash{\(j_k\)}}}%
    \put(0.34427214,0.12391882){\color[rgb]{0,0,0}\makebox(0,0)[lb]{\smash{\(j_{ik}\)}}}%
    \put(0.26541649,0.21046771){\color[rgb]{0,0,0}\makebox(0,0)[lb]{\smash{\(\bigsqcup_{\alpha \in A} M_\alpha\)}}}%
    \put(0.26733982,0.05275646){\color[rgb]{0,0,0}\makebox(0,0)[lb]{\smash{\(M_{\alpha_i} \sqcup M_{\alpha_k}\)}}}%
    \put(0.50967663,0.02390681){\color[rgb]{0,0,0}\makebox(0,0)[lb]{\smash{\(\pi\)}}}%
    \put(0.50198343,0.2470106){\color[rgb]{0,0,0}\makebox(0,0)[lb]{\smash{\(\pi\)}}}%
    \put(0.45891217,0.12921649){\color[rgb]{0,0,0}\rotatebox{20.24350861}{\makebox(0,0)[lb]{\smash{\(\pi \circ j_{ik}\)}}}}%
    \put(0.64623149,0.20662109){\color[rgb]{0,0,0}\makebox(0,0)[lb]{\smash{\(\big(\bigsqcup_{\alpha \in A} M_\alpha\big)/_\sim\)}}}%
    \put(0.65584806,0.04890987){\color[rgb]{0,0,0}\makebox(0,0)[lb]{\smash{\(\big(M_{\alpha_i} \sqcup M_{\alpha_k}\big)/_\sim\)}}}%
    \put(0.73085709,0.12007224){\color[rgb]{0,0,0}\makebox(0,0)[lb]{\smash{\(\tilde{j}_{ik}\)}}}%
  \end{picture}%
\endgroup%

\end{center}

Here, all \(j's\) denote canonical inclusion maps (in particular $j_i$ and $j_k$ denote the inclusion maps of $M_{\alpha_i}$ and $M_{\alpha_k}$ into $\bigsqcup_{\alpha \in A} M_\alpha$), the \(\pi's\) denote projection maps, the lower equivalence relation is defined as in the proof of Theorem \ref{CommonExtension} and it is easy to check that the map \(\pi \circ j_{ik}\) descends to the quotient, i.e.\ to \(\tilde{j}_{ik}\).

As for \eqref{homeo} one checks that \(\pi \circ j_{ik}\) is an open map. Thus, \(\tilde{j}_{ik}\) is open as well. Since \(\tilde{j}_{ik}\) is also continuous and injective, it is a homeomorphism onto its image.

In Theorem \ref{CommonExtension} we proved that the quotient topology on \((M_{\alpha_i} \sqcup M_{\alpha_k})/_\sim\) is Hausdorff - thus we can find open neighbourhoods that separate \([p_{\alpha_i}]\) and \([q_{\alpha_k}]\) in \((M_{\alpha_i} \sqcup M_{\alpha_k})/_\sim\). Pushing forward these neighbourhoods to \((\bigsqcup_{\alpha \in A} M_\alpha)/_\sim\) via \(\tilde{j}_{ik}\) we obtain separating open neighbourhoods of \([p_{\alpha_i}]\) and \([q_{\alpha_k}]\) in \(\tilde{M}\).

\item \emph{\(\tilde{M}\) is locally euclidean and has a natural smooth structure:} This is seen exactly as in the proof of Theorem \ref{CommonExtension}.
\item \emph{\(\tilde{M}\) has a natural smooth Lorentzian metric that is Ricci-flat and comes with a natural time orientation:}
Again, this is seen exactly as before.
\item \emph{\((\tilde{M}, \tilde{g})\) is globally hyperbolic with Cauchy surface \(\tilde{\iota}(\overline{M})\):} Here, \(\tilde{\iota} := \pi \circ j_i \circ \iota_i\) for some \(\alpha_i \in A\). This definition does obviously not depend on  \(\alpha_i \in A\). 

The proof is also nearly the same as before. Let \(\gamma : I \to \tilde{M}\) be an inextendible timelike curve. For \(t_0 \in I\) we have, say, \(\gamma(t_0) \in (\pi \circ j_i) ( M_{\alpha_i})\). Let \(J \ni t_0\)  denote the maximal connected subinterval of \(I\) such that \(\gamma(J) \subseteq (\pi \circ j_i) (M_{\alpha_i})\). We can then pull back \(\gamma|_J\)  via \(\pi \circ j_i\) to \(M_{\alpha_i}\), which gives rise to an inextendible timelike curve in \(M_{\alpha_i}\) that has to intersect \(\iota_i(\overline{M})\). Thus \(\gamma\) intersects \(\tilde{\iota}(\overline{M})\).

Assume \(\gamma\) intersected \(\tilde{\iota}(\overline{M})\) more than once. Again, we can find \(t_1 < t_4 \in I\) with \(\gamma(t_1), \gamma(t_4) \in \tilde{\iota}(\overline{M})\) and \(\gamma(t) \notin \tilde{\iota}(\overline{M})\) for \(t_1 < t <t_4\). Since \(\gamma\) is continuous and \([t_1,t_4]\) is compact, \(\gamma([t_1,t_4])\) is contained in finitely many \(\pi \circ j_\alpha (M_{\alpha})\). But since each of these \({M_\alpha}'s\) is globally hyperbolic one can actually reduce this cover to just two elements, since otherwise one would get an inextendible timelike curve of the form \(\gamma|_{[t_2,t_3]}\) in some \(M_\alpha\), where \(t_1 < t_2 < t_3 < t_4\), that does not intersect \(\iota_\alpha(\overline{M})\).

From here on, one follows the remaining argument from point \(4\) of the proof of Theorem \ref{CommonExtension}.

\item \emph{\(\tilde{M}\) is second countable:} This follows directly from a Theorem of Geroch, see the appendix of \cite{Ger68}, where he shows that any manifold that is connected\footnote{That \(\tilde{M}\) is connected here follows trivially from it being globally hyperbolic, hence path connected (recall that we assumed that \(\overline{M}\) is connected).}, Hausdorff and locally euclidean and which, moreover, admits a smooth Lorentzian metric, is also second countable.

\item \emph{\(\tilde{M}\) is an extension of any GHD of the same initial data:} Let \((M, g, \iota)\) be a GHD of the same initial data. Since \(M\) is second countable and time oriented, we can find a globally timelike vector field \(T\) on \(M\).  Let us denote with \(I_x \subseteq \R\) the maximal time interval of existence of the integral curve of \(T\) starting at \(x \in M\).\footnote{Note that the existence of such a maximal time interval follows from an elementary `taking the union of all time intervals of existence argument' - without appealing to Zorn's lemma.} In the following we recall some results from standard ODE theory:
The set \(\mathcal{D} := \{ (x,t) \in M \times \R \; | \; t \in I_x \}\) is open and the flow \(\Phi : \mathcal{D} \to M\) of \(T\) is smooth. Moreover, if we fix \(t \in \R\) and regard \(\Phi_t (\cdot) := \Phi\big((\cdot,t)\big)\) as a function from some open subset of \(M\) to \(M\), then \(\Phi_t\) is a local diffeomorphism.

We now define \(\mathcal{D}_{\iota(\overline{M})}:= \{ (x,t) \in \iota(\overline{M}) \times \R \; | \; t \in I_x\}\), which is an open neighbourhood of \(\iota(\overline{M}) \times \{0\}\) in \(\iota(\overline{M}) \times \R\) (again by standard ODE theory), and claim that \(\chi := \Phi\big|_{\mathcal{D}_{\iota(\overline{M})}} : \mathcal{D}_{\iota(\overline{M})} \to M\) is a diffeomorphism. 

The smoothness of \(\chi\) follows directly from the smoothness of \(\Phi\), and the bijectivity follows from the global hyperbolicity of \(M\). More precisely, since every maximal integral curve of \(T\) (which is, in particular, an inextendible timelike curve) has to intersect \(\iota(\overline{M})\), \(\chi\) is surjective; and since every such curve intersects \(\iota(\overline{M})\) exactly once, we obtain the injectivity. In order to see that \(\chi\) is a local diffeomorphism, let \((x,t) \in \mathcal{D}_{\iota(\overline{M})}\) and choose a basis \((Z_1, \ldots, Z_d)\) of \(T_x \iota(\overline{M})\). We have
\begin{equation}
\label{small}
\chi_*\big|_{(x,t)}(Z_i) = \big(\Phi_t\big)_*\big|_x(Z_i) \qquad \textnormal{ and } \qquad \chi_*\big|_{(x,t)} (\partial_t) = T\big|_{\Phi_t(x)} = \big(\Phi_t\big)_*\big|_x (T\big|_x) \;.
\end{equation}
Since \(\iota(\overline{M})\) is spacelike, \((Z_1, \ldots, Z_d, T_x)\) forms a basis for \(T_xM\); and since \(\Phi_t\) is a local diffeomorphism, it follows from \eqref{small} that \(\chi_*\) is surjective. Thus, we have shown that \(\chi\) is a diffeomorphism.

It now follows that \(\chi \circ (\iota \times \mathrm{id})\) is a diffeomorphism from some open neighbourhood of \(\overline{M} \times \{0\}\) in \(\overline{M} \times \R\) to \(M\) which maps \(\overline{M} \times \{0\}\) on \(\iota(\overline{M})\). Pulling back the Lorentzian metric, we obtain that there is an \(M_{\alpha_i} \in X\) that is isometric to \(M\) via \(\chi \circ (\iota \times \mathrm{id})\). The isometric embedding of \(M\) into \(\tilde{M}\) is now given by \(\pi \circ j_i \circ \big(\chi \circ (\iota \times \mathrm{id})\big)^{-1}\).

Finally, it is straightforward to deduce from this maximality property that \(\tilde{M}\) is, up to isometry, the only GHD with this property.
\end{enumerate}
This finally finishes the proof of the existence of the MGHD.
\end{proof}

\section*{Acknowledgements}
I would like to thank my supervisor Mihalis Dafermos not only for bringing to my attention that it might be of interest to find a `dezornification' of the proof of the existence of a maximal GHD, but also for numerous instructive discussions and helpful comments on a previous version of this paper. Here, I would also like to thank Marc Nardmann for valuable feedback on a previous version of this manuscript, and especially for bringing to my attention that the collection of all GHDs forms a proper class. Moreover, I am grateful to John Conway for a very stimulating discussion about the axiom of choice. Finally, I would like to thank the Science and Technology Facilities Council (STFC) and the German Academic Exchange Service (DAAD) (Doktorandenstipendium) for their financial support.

\bibliographystyle{acm}
\bibliography{Bibly}

\end{document}